\documentclass[review]{elsarticle}

\usepackage{url}
\usepackage{float}
\usepackage{boldline}
\usepackage{booktabs}
\usepackage{amssymb,amsmath,amsthm}
\newtheorem{mydef}{Definition}
\newtheorem{mylemma}{Lemma}
\newtheorem{myprop}{Proposition}
\newtheorem{mytheorem}{Theorem}
\usepackage{multirow}
\usepackage{hyperref}
\hypersetup{
    colorlinks=true,
    linkcolor=magenta,
    filecolor=blue,      
    urlcolor=blue,
    citecolor=blue,
    pdftitle={Sharelatex Example},
    bookmarks=true,
    }
\usepackage{adjustbox}
\usepackage{mathtools}
\DeclarePairedDelimiter\ceil{\lceil}{\rceil}
\usepackage[linesnumbered,ruled,vlined]{algorithm2e}
\usepackage[framemethod=TikZ]{mdframed}
\mdfdefinestyle{MyFrame}{%
	linecolor=blue,
	outerlinewidth=1.0pt,
	roundcorner=20pt,
	innertopmargin=\baselineskip,
	innerbottommargin=\baselineskip,
	innerrightmargin=20pt,
	innerleftmargin=20pt,
	backgroundcolor=gray!50!white}
\usepackage{tcolorbox}

\journal{Journal of \LaTeX\ Templates}

\biboptions{authoryear}

\begin{document}

\begin{frontmatter}

\title{A Co\mbox{-}Operative Game Theoretic Approach for the Budgeted Influence Maximization Problem
}

\author{Suman Banerjee\fnref{1234}}
\address{Department of Computer Science and Engineering,\\
Indian Institute of Technology, Jammu \\
Jammu \& Kashmir-181221, India.\\
\ead{suman.banerjee@iitjammu.ac.in}}
\cortext[mycorrespondingauthor]{Corresponding author}


%


\begin{abstract}
Given a social network of users with selection cost, the \textsc{Budgeted Influence Maximization Problem} (\emph{BIM Problem} in short) asks for selecting a subset of the nodes (known as \emph{seed nodes}) within an allocated budget for initial activation to maximize the influence in the network. In this paper, we study this problem under the \emph{co\mbox{-}operative game theoretic} framework. We model this problem as a co\mbox{-}operative game where the users of the network are the players and for a group of users, the expected influence by them under the \emph{Maximum Influence Arborences} diffusion model is its utility. We call this game as \emph{BIM Game} and show this is `non-convex' and `sub-additive'. Based on the proposed game\mbox{-}theoretic model and using the solution concept called `Shapley Value', we propose an iterative algorithm for finding seed nodes. The proposed methodology is divided into mainly two broad steps: the first one is computing the approximate marginal gain in \emph{Shapley Value} for all the nodes of the network, and the second one is selecting seed nodes from the sorted list until the budget is exhausted. We also show that the proposed methodology can even be more effective when the community structure of the network is exploited. The proposed methodologies have been implemented, and an extensive set of experiments have been conducted with three publicly available social network datasets. From the experiments, we observe that the seed set selected by the proposed methodologies lead to more number of influence nodes compared to many standard and baseline methods from the literature with a reasonable computational overhead. In particular, if the community structure of the network is exploited then there is an increase upto $2 \%$ in number of influenced nodes.
\end{abstract}


\end{frontmatter}


\section{Introduction}
The \emph{Influence Maximization Problem} aims at identifying a small set of highly influential users such that their initial activation leads to the maximum number of influential nodes in the network~ \cite{kempe2003maximizing}. This is perhaps one of the most well studied problem in computational social network analysis due to its practical applications in different domains including \emph{crowed\mbox{-}sourcing}~ \cite{hossain2012users}, \emph{viral marketing}~ \cite{domingos2005mining}, \emph{computational advertisement}~ \cite{aslay2015viral}, \emph{recommender systems}~ \cite{ye2012exploring} etc. This problem was initially studied by ~\cite{kempe2003maximizing} and since then there was an extensive effort to study in and around of this problem. Look into \cite{banerjee2020survey, li2018influence} for recent surveys on this topic. The main cause of influence is the diffusion of information and it is basically a cascading effect by which information propagates from one part of the network to the other. To study this process, several diffusion models have been introduced, such as \emph{Independent Cascade Model (IC Model)}, \emph{Linear Threshold Model (LT Model)}, \emph{Maximum Influence Arborences Model (MIA Model)} and many more~\cite{li2017survey}.
\par In reality, social networks are formed by rational human agents. This means that if an user is selected as a seed then incentivization is required. However, the basic influence maximization problem assumes that the every users of the network has equal selection cost. though it may not be so in reality. To bridge this gap,~\cite{nguyen2013budgeted} introduced the  `Budgeted Influence Maximization Problem', where the users of the  social network are associated with a nonuniform selection cost and a fixed budget is allocated, the goal here is to select a subset of nodes within the allocated budget such that their initial activation leads to maximum number of influenced nodes. As compared to the basic influence maximization problem, existing literature of this problem is very very limited. In this paper, we study this problem under the co-operative game theoretic framework.
\par \emph{Game Theory}, which is basically a mathematical study of strategic interaction among a group of rational agents, has been used to solve many problems in the domain of social network analysis such as \emph{community detection}~\cite{chen2010game}, \emph{opinion dynamics}~ \cite{ding2010evolutionary}, \emph{leader selection}~ \cite{zimmermann2005cooperation}, \emph{rumor dissemination}~ \cite{kostka2008word}, \emph{influence maximization}~ \cite{clark2011maximizing} and many more. However, to the best of our knowledge, the BIM Problem has not been studied yet under the game theoretic framework. In this paper, we study the BIM Problem under  co\mbox{-}operative game theoretic framework and our study is motivated by the work  \cite{narayanam2010shapley}. However, there are the following fundamental differences:
\begin{itemize}
\item ~\cite{narayanam2010shapley}'s study is concentrated on the basic influence maximization problem and the non\mbox{-}uniform nature of the selection cost has not been taken into account, which is a practical concern.
\item In their study they have used IC and LT as the underlying diffusion model. However, recently there are several studies on influence maximization that considers MIA Model as the underlying model of diffusion \cite{ke2018finding}. In this study also, we have used the MIA as the underlying diffusion model.
\item It is well known that the complexity of computing the shapley value of a $n$ player co-operative game is of $\mathcal{O}(n!)$ ~\cite{narahari2014game}. Even in case of a small size social network (e.g., number of nodes are $100$) this is a huge computational burden. To get rid of this, they have randomly chosen linear number of permutations of the players for computing shapley value~\cite{narayanam2010shapley}. However, with the help of an existing result by ~\cite{maleki2013bounding}, we use the appropriate number of samples such that shapley value can be computed with bounded error with high probability.
\end{itemize}
In particular, we make the following contributions in this paper: 
\begin{itemize}
\item We formulate a co\mbox{-}operative diffusion game for the \textsc{Budgeted Influence Maximization Problem}, and design an iterative algorithm based on the solution concept of a co-operative game known as Shapley value for identifying the influential nodes. 
\item We study the important properties of the formulated game and detailed complexity analysis of the proposed methodologies have also been done.
\item We also show that if we consider the community structure of the network, then the proposed methodology can leads to even more number of influenced nodes.
\item The proposed methodologies have been with four publicly available social network datasets and an extensive set of experiments have been carried out to demonstrate the effectiveness of the proposed methodology.  
\end{itemize}
Rest of the paper is organized as follows: Section \ref{Sec:RW} describes relevant studies from the literature. Section \ref{Sec:Prelim and Prob} describes required preliminary concepts and define the diffusion game formally. The proposed solution methodologies have been described in Section \ref{Sec:Methodology}. Section \ref{Sec:EE} contains the experimental evaluation, and finally, Section \ref{Sec:CFD} concludes our study and gives future directions.

\section{Related Work} \label{Sec:RW}
 Our study is closely related to the `Influence Maximization in Social Networks', more particularly the Budgeted Influence Maximization Problem, and also game theoretic solution methodologies for the influence maximization problem. Here, we present relevant studies from the literature. 
\paragraph{\textbf{Influence Maximization in Social Networks}} The problem of influence maximization aims at choosing a small number of highly influential users in a social network for initial activation such that due to the cascading effect of diffusion, the number of influenced node is maximized ~\cite{kempe2003maximizing}. ~\cite{domingos2001mining, richardson2002mining} first introduced this problem for the `viral marketing' in social networks. Later on, ~\cite{kempe2003maximizing, kempe2005influential, DBLP:journals/toc/KempeKT15} studied the computational issues of this problem and showed that the problem is $\mathcal{NP}$\mbox{-}Hard under the \emph{Independent Cascade} and \emph{Linear Threshold} Model of diffusion. However, they gave a $(1-\frac{1}{e})$\mbox{-} factor approximation algorithm, which works based on maximum marginal influence gain computation and $\mathcal{O}(n^{1-\epsilon})$ for any $\epsilon > 0$ factor inapproximability result. Their study remains an influential one and triggers a huge amount of study in last one and half decades. Solution methodologies can be grouped into different categories, such as `approximation algorithms' such as \emph{CELF} ~\cite{leskovec2007cost}, \emph{CELF++} ~\cite{goyal2011celf}, \emph{MIA}~ \cite{wang2012scalable}, \emph{PMIA} ~\cite{wang2012scalable}, \emph{TIM} ~\cite{tang2014influence}, \emph{IMM} ~\cite{tang2015influence}; heuristic solutions such as \emph{SIMPATH} ~\cite{goyal2011simpath}, \emph{IRIE} ~\cite{jung2012irie}, LDGA ~\cite{chen2010scalable}, \emph{Community\mbox{-}Based Solution Methodologies} such as \emph{CIM} ~\cite{chen2012exploring, chen2014cim}, \emph{ComPath} \cite{rahimkhani2015fast}, \emph{INCIM} \cite{bozorgi2016incim}, \emph{CoFIM} \cite{shang2017cofim} and many more. Please refer to \cite{banerjee2020survey, hafiene2020influential,  li2018influence} (and references therein) for recent surveys.
\paragraph{\textbf{Budgeted Influence Maximization Problem}} In case of BIM Problem, along with the input social network, we are also given with a cost function that assigns selection cost to each node and a fixed budget is allocated for the seed set selection. \cite{nguyen2013budgeted} first introduced the BIM Problem and proposed a $(1-\frac{1}{\sqrt{e}})$\mbox{-}factor approximation algorithm and two efficient heuristics to solve this problem. Recently,  \cite{wang2020efficient} studied the BIM Problem and proposed a solution methodology that gives an approximation ratio of $\frac{1}{2}(1-\frac{1}{e})$. They further showed that this can be improved upto $(1-\beta)(1-\frac{1}{e})$. \cite{guney2019optimal} proposed an integer programming\mbox{-}based approach to solve this problem under the IC Model of diffusion.  \cite{han2014balanced} proposed a couple of heuristics for this problem that carefully considers both cost effective nodes and influential nodes. Recently, \cite{banerjee2019combim} proposed a community\mbox{-}based solution methodology for the BIM Problem which is broadly divided into four steps, namely, community detection, budget distribution, seed node selection, and budget transfer.  \cite{shi2019post} proposed two different solution methodologies with data dependent approximation ratio. \cite{yu2018fast} studied this problem under \emph{credit distribution} model and came up with a streaming algorithm with $(\frac{1}{3} - \epsilon)$ approximation of the optimum.  
\paragraph{\textbf{Game Theoretic Solution Methodologies for SIM and Related Problems}} Game theoretic techniques have been used to solve the  influence maximization problems since last one decade. To the best of our knowledge, the first study in this direction was by \cite{narayanam2010shapley}. They proposed a `Shapley Value'-based approach for selecting the seed nodes for the influential maximization problem. \cite{clark2011maximizing} studied the influence maximization problem in competitive situation and they formulated a \emph{Stackleberg Game} and proposed a methodology to solve the game. \cite{borodin2010threshold} studied the influence maximization problem in competitive situation with several natural extensions of the LT Model. \cite{angell2017don} studied the influence maximization problem and proposed a greedy heuristic that shows an improvement of $7 \%$ and $55 \%$ for submodular and non\mbox{-}submodular influence functions, respectively.  \cite{li2015getreal} study the influence maximization problem in competitive situation and proposed a game theoretic framework for this problem. They designed a non\mbox{-}cooperative game and proposed an algorithm that computes the `Nash Equilibrium' of the game that guarantees optimal strategies. Recently, \cite{wang2020fast} proposed an  cumulative oversampling technique for `Thompson Sampling' to construct optimistic parameter estimates with fewer samples. They showed that their learning algorithms can be used to solve the BIM Problem using less number of samples.
\par To the best of our knowledge, there does not exist any literature that studies BIM Problem under game theoretic setting. In this paper, we study this  Problem under Co\mbox{-}Operative Game Theoretic Framework. In particular, we propose an iterative approach based on the Shapley Value for selecting the influential nodes for the BIM Problem. We also show that the exploitation of community structure helps the proposed methodology to achieve even better influence spread.

\section{Preliminary Concepts, Background, and Problem Definition} \label{Sec:Prelim and Prob}
Here, we present some preliminary definitions. Initially, we start with the  social network specific ones.
\subsection{Social Network Specific}
\begin{mydef}[Social Network]
A social network is an interconnected structure among a group of human agents formed for social interactions, which is often represented as a  graph $G(V, E, P)$. Here, the vertex set $V(G)=\{u_1, u_2, \dots, u_n\}$ are the set of $n$ users, the edge set $E(G)=\{e_1, e_2, \dots, e_m\}$ are the $m$ social ties among the users, and $P: E(G) \rightarrow (0,1]$ is the edge weight function that assigns each edge to its influence probability.
\end{mydef}
We denote the number of nodes and edges of the network by $n$ and $m$, respectively. For any positive integer $n$, $[n]$ denotes the set $\{1,2, \cdots, n\}$. For any $e \equiv(uv) \in E(G)$, we denote the influence probability assigned to it as $P(e)$ or $P_{uv}$. If, $(uv) \notin E(G)$ then $P_{uv}=0$. Now, to start a diffusion process in the network, there should be some initial adopters, which we call as \emph{seed nodes}. To study the diffusion process in a network, several models have been introduced such as IC Model, LT Model, MIA Model and so on \cite{guille2013information}. Now, once a diffusion process is initiated from a number of seed nodes, it ends with influencing a subset of the nodes and called as the influence of the seed set, which is defined next.
\begin{mydef}[Influence of a Seed Set]
Given a social network $G(V, E, P)$ and a diffusion model $\mathcal{M}$ for a given seed set $S \subseteq V(G)$, its influence is defined as the number of nodes that are influenced due to the initial activation of the nodes in $S$ if the information is diffused by the rule of the model $\mathcal{M}$. $I_{\mathcal{M}}(S)$ denotes the the set of influenced nodes by the seed set $S$ under the diffusion model $\mathcal{M}$. This quantity is measured in terms of expactation. Hence, the influence of the seed set $S$ and under the diffusion model $\mathcal{M}$ is $\sigma_{\mathcal{M}}(S)=\mathbb{E}(|I_{\mathcal{M}}(S)|)$, where $\sigma(.)$ is the social influence function, i.e., $\sigma: 2^{V(G)} \rightarrow \mathbb{R}_{0}^{+}$ with $\sigma(\emptyset)=0$. Here, $\emptyset$ denotes the empty set.
\end{mydef}
Now, any real\mbox{-}world diffusion processes (such as political campaigns, viral marketing of products etc.) are always carried out to maximize the influence. In this direction, next we state the well studied Social Influence Maximization Problem.
\begin{mydef}[Social Influence Maximization Problem (SIM Problem)]\cite{kempe2003maximizing}
Given a social network $G(V, E, P)$, a diffusion model $\mathcal{M}$, and a positive integer ($k \in \mathbb{Z}^{+}$), the social influence maximization problem asks for selecting a seed set $S \subseteq V(G)$ with $|S| \leq k$, whose initial activation leads to the maximum number of influenced nodes, if the diffusion process happens by the rule of the model $\mathcal{M}$. Mathematically,
\begin{equation}
S^{*}=\underset{ S \subseteq V(G), |S| \leq k}{argmax} \  \sigma_{\mathcal{M}}(S)
\end{equation}
\end{mydef}
As mentioned previously, real\mbox{-}world social networks are formed by rational and self interested agents. If a node is selected as seed then some kind of incentivization is required. In reality, different users of the network have different selection cost, i.e., users have non uniform selection cost. However, the SIM Problem assumes that the selection costs are uniform. To bridge this gap, the `Budgeted Influence Maximization Problem' has been introduced by \cite{nguyen2013budgeted} which is defined next.
\begin{mydef}[\textsc{Budgeted Influence Maximization Problem}]
Given a social network $G(V, E, P)$ along with a cost function $C: V(G) \rightarrow \mathbb{Z}^{+}$ that assigns each node to its selection cost, a diffusion model $\mathcal{M}$, and a fixed budget $B$, the problem of budgeted influence maximization asks for selecting a seed set within the allocated budget, whose initial activation leads to the  maximum number of influenced nodes, if the diffusion process happens by the rule of the model $\mathcal{M}$. Mathematically,
\begin{equation}
S^{*}=\underset{S \subseteq V(G), C(S) \leq B}{argmax} \  \sigma_{\mathcal{M}}(S),
\end{equation}
where, $C(S)$ denotes the total selection cost of the users in $\mathcal{S}$, i.e., $C(S)=\underset{u \in S}{\sum} C(u)$. We denote an arbitrary instance of BIM Problem as $<G(V, E, P),C,B>$. 
\end{mydef}
From the algorithmic point of view the problem can be posed as follows:\\
\begin{center}
\begin{tcolorbox}[title=\textsc{Budgeted Influence Maximization Problem}, width=12cm]
\textbf{Input:} Social Network $G(V,E,P)$, Cost Function $C: V(G) \rightarrow \mathbb{Z}^{+}$, Budget $B$, and the Diffusion Model $M$.

\textbf{Problem:} Find out a seed set $S \subseteq V(G)$, such that $C(S) \leq B$ and $\sigma_{M}(S)$ is maximized.
\end{tcolorbox}
\end{center}

Next, we state the \emph{Maximum Influence Aroborance (MIA)} diffusion model, which we have considered as the underlying diffusion model. Symbols and notations used in this paper is given in Table \ref{Tab:Symbols}.
\begin{table}[h!]
  \begin{center}
    \caption{Symbols and notations used in this study}
    \label{Tab:Symbols}
    \begin{tabular}{|l|c|} 
    \hline
      \textbf{Symbol} & \textbf{Meaning} \\
      \hline
      $G(V,E,P)$& The input social network \\
      \hline
      $V(G)$, $E(G)$ & Vertex and Edge set of $G$ \\
      \hline
      $n,m$ & The number vertices and edges of $G$\\
      \hline 
      $P$ & The edge probability function \\
      \hline
      $K$ & The set of communities of $G$\\
      \hline
      $\ell$ & The number of communities of $G$\\
      \hline
      $P(e)/P_{uv}$ & The edge probability of the edge $e\equiv(uv)$\\
      \hline
      $\mathcal{P}(p)$ & Propagation probability of the path $p$\\
      \hline
      $C$ & The cost function\\
      \hline
      $C(u)$ & Cost of the user $u$\\
      \hline
      $B$ & Budget for seed set selection\\
      \hline
      $\sigma(.)$ & The social influence function\\
      \hline
      $S$ & The seed set\\
      \hline
      $C(S)$ & Total selection cost of the users in $S$\\
      \hline
      $\sigma(S)$ & The influence of the seed set $S$\\
      \hline
      $p_{u_iu_j}^{max}$ & Maximum influence path from $u_i$ to $u_j$\\
      \hline
      $\mathbb{P}(u_i,u_j)$ & Set of all paths between $u_i$ and $u_j$\\
      \hline
      $\theta$ & Cut Off Probability \\
      \hline
      $d_{max}$ & Maximum degree of $G$\\
      \hline
      $N$ & The set of players $\{1, 2, \ldots, n\}$\\
      \hline
      $\nu$ & The utility/ pay-off function\\
      \hline
      $(N, \nu)$ & The BIM Game \\
      \hline
      $\Psi_{i}(\mathcal{N}, \nu)$ & Shapley value of the $i$\mbox{-}th player\\
      \hline
      $MIIA(v, \theta)$ & The set of MIA Paths to the node $v$\\
      \hline
      $\mathbb{Z}^{+}$ & The set of positive integers\\
      \hline
      $\mathbb{R}_{0}^{+}$ & The set of positive real numbers including zero\\
      \hline
    \end{tabular}
  \end{center}
\end{table}

\subsection{The MIA Diffusion Model}
This is the diffusion model introduced by \cite{wang2012scalable}. Before stating the diffusion model first we state two preliminary definitions.
\begin{mydef}[Propagation Probability of a Path]
Given two vertices $u_i, u_j \in V(G)$, let $\mathbb{P}(u_i,u_j)$  denotes the set of paths from the vertex $u_i$ to $u_j$. For any arbitrary path $p \in \mathbb{P}(u_i,u_j)$ the propagation probability  is defined as the product of the influence probabilities of the edges that constitute the path.
\begin{equation}
  \mathcal{P}(p) =
  \begin{cases}
    \underset{(u_iu_j) \in p}{\prod} P_{u_iu_j} & \text{if $\mathbb{P}(u_i,u_j) \neq \phi$} \\
    0 & \text{otherwise}\\
  \end{cases}
\end{equation}
\end{mydef}

\begin{mydef}[Maximum Probabilistic Path]
Given two vertices $u_i,u_j \in V(\mathcal{G})$, the maximum probabilistic path is the path with the maximum propagation probability and denoted as $p_{u_i,u_j}^{max}$. Hence,
\begin{equation}
p_{u_i,u_j}^{max}=\underset{p \in \mathbb{P}(u_i,u_j)}{argmax} \ \mathcal{P}(p)
\end{equation}
\end{mydef}
Based on the path propagation probability and and maximum probabilistic path maximum influence in-arborescences are defined as follows.
\begin{mydef}[Maximum Influence In-Arborescences (MIIA)]
For a node $v \in V(G)$, and a probability threshold $\theta$, the maximum influence in-arborescence is the union of the maximum influence paths with more than the cut off probability $\theta$ to $v$. Mathematically,
\begin{equation}
MIIA(v, \theta)=\underset{u \in V(G) \setminus v; p_{u,v}^{max} \geq \theta}{\bigcup} p_{u,v}^{max}
\end{equation}
\end{mydef} 

Now, given a seed set $S$, a node $v \in V(G) \setminus S$ and its $MIIA(v,\theta)$, by the rule of $MIA$ model it is assumed that the influence from $S$ to $v$ is propagated through the paths in $MIIA(v,\theta)$. Let $ap(u, S, MIIA(v, \theta))$ be the influence probability of the node $u$ by the seed set $S$ and this can be recursively computed as mentioned in \cite{wang2012scalable}.

\begin{mydef}[MIA Diffusion Model]
For any seed set $S \subseteq V(G)$ and any arbitrary node $v \in V(G) \setminus S$, in the MIA Model, it is assumed that the nodes in $S$ will influence $v$ through the paths in $MIIA(v, \theta)$ and the expected influence by the seed set $S$ can be given by the Equation \ref{Eq:Inf}
\begin{equation} \label{Eq:Inf}
\sigma(S)= \underset{v \in V(G)}{\sum} ap(v,S, MIIA(v,\theta))
\end{equation}
\end{mydef}
The following two lemmas will be useful to prove some properties of the diffusion game, which will be defined later.
\begin{mylemma} \cite{wang2012scalable} \label{Lemma:1}
The influence function mentioned in Equation \ref{Eq:Inf} is non\mbox{-}negative, and monotone (i.e., $\forall S \subseteq V(G)$, $\sigma(S) \geq 0$ and $\forall S \subseteq T \subseteq V(G)$, $\sigma(S) \leq \sigma(T)$). 
\end{mylemma}

\begin{mylemma} \cite{wang2012scalable} \label{Lemma:2}
The influence function mentioned in Equation \ref{Eq:Inf} is submodular (i.e., $\forall S \subseteq T \subseteq V(G)$, and $u \in V(G) \setminus T$ and $ \sigma(S \cup \{u\})- \sigma(S) \geq \sigma(T \cup \{u\})- \sigma(S)$. 
\end{mylemma}

\subsection{Co\mbox{-}Operative Game Theory Specific}
Co-operative game theory is the study about strategic formation of coalition and their mathematical analysis which is defined next. 
\begin{mydef}[Co\mbox{-}Operative Game or Coalition Game or Transferable Utility Game]
A Co\mbox{-}Operative Game is defined by the tuple $(\mathcal{N}, \nu)$, where $\mathcal{N}=\{1, 2, \dots, n\}$ is the finite set of players and $\nu$ is called the payoff function (also known as the characteristic function) that assigns each possible coalition to a real number, i.e., $\nu: 2^\mathcal{N} \rightarrow \mathbb{R}^{+}$. It is assumed that $\nu(\emptyset)=0$.
\end{mydef}
In co\mbox{-}operative game theory, one of the main concern is how to distribute the total utility among the players. It is called as the \emph{payoff allocation}.
\begin{mydef}[Payoff Allocation]
A payoff allocation is a vector $v=(v_i)_{i \in \mathcal{N}}$ in $\mathbb{R}^{\mathcal{N}}$, where each entry represents the utility share to the corresponding player.
\end{mydef}
 The \emph{Shapley Value} is an important solution concept in co-operative game theory, which performs the payoff allocation satisfying the following three properties, namely, symmetry, linearity, and carrier \cite{narahari2014game}.
\begin{mydef}[Shapley Value]
In a co-operative game $(\mathcal{N}, \nu)$, the Shapley Value of the player $i$ towards a coalition $\mathcal{T}$ is defined as 
\begin{equation}
\Psi_{i}(\mathcal{N}, \nu)=\underset{\mathcal{T} \subseteq \mathcal{N} \setminus \{i\}}{\sum} \frac{|\mathcal{T}|! * (|\mathcal{N}| - |\mathcal{T}|-1)!}{|\mathcal{N}|!} (\nu(\mathcal{T} \cup \{i\})- \nu(\mathcal{T}))
\end{equation}
\end{mydef}

\section{Proposed Methodology} \label{Sec:Methodology}
In this section, we describe the game theoretic approach for solving the BIM Problem. This section is broadly divided into the following  subsections. In Sub-section \ref{SubSec:BIM_Game}, we define the BIM Game and establish its properties. Sub-section \ref{SubSec:Overview} contains the overview of the proposed methodology. Sub-section \ref{SubSec:Algo} contains the algorithms present in the proposed methodology, and finally, Sub-section \ref{SubSec:Complexity} contains time and space complexity analysis of our proposed approaches.
\subsection{The BIM Game and Its Properties} \label{SubSec:BIM_Game}
\begin{mydef}[The BIM Game] \label{Def:Game}
We define our diffusion game as a co\mbox{-}operative game, where the nodes of the network are the players, and for any subset of players, their utility is the expected influence in the network under the MIA Model of diffusion. We denote this diffusion game as $(\mathcal{N}, \nu)$, where $\mathcal{N}=V(G)=\{u_1, u_2, \ldots, u_n\}$,  $\nu: 2^{V(G)} \longrightarrow \mathbb{R}^{+}$, and for any $S \subseteq V(G)$, $\nu(S)=\sigma(S)$ \footnote{As mentioned previously, in this study we assume that the diffusion is happening by the rule of MIA Model. Hence, now onwards we omit the subscript $\mathcal{M}$ from $\sigma(.)$}. 
\end{mydef}
 Now, we show some of the important properties of the utility function and the proposed diffusion game.
\begin{myprop}[Non\mbox{-}negativity and Monotonicity of $\nu(.)$]
The utility function of the diffusion game in Definition \ref{Def:Game} is non\mbox{-}negative and monotone.
\end{myprop}
\begin{proof}
As mentioned in Lemma \ref{Lemma:1}, under MIA diffusion model, the influence function is non\mbox{-}negativity and monotone, the same holds for the utility function as well.
\end{proof}
\begin{myprop}[Non\mbox{-}Convexity of the Diffusion Game]
The BIM Game defined in Definition \ref{Def:Game} is not convex.
\end{myprop}
\begin{proof}
By definition, a co\mbox{-}operative game is said to be convex, if its utility function $\nu(.)$ has the following property: for all $S \subseteq T \subseteq \mathcal{N}$ and for all $u \in \mathcal{N} \setminus T$, $\nu(S \cup \{u\})- \nu(S) \leq \nu(T \cup \{u\})- \nu(T)$. As mentioned in Lemma \ref{Lemma:2}, the influence function $\sigma(.)$ is submodular, this implies that the diffusion game can not be convex.
\end{proof}
\begin{myprop}[Sub\mbox{-}Additivity of the BIM Game]
The BIM Game mentioned in Definition \ref{Def:Game} is sub\mbox{-}additive.
\end{myprop}
\begin{proof}
Let, $\mathcal{U}_{1}$ and $\mathcal{U}_{2}$ be two coalition of the BIM Game $(N, \nu)$. Now, we need to show that the utility of the larger coalition, i.e., $\mathcal{U}_{1} \cup \mathcal{U}_{2}$ is at most as the total utility of the individual coalition; i.e.; $\nu(\mathcal{U}_{1} \cup \mathcal{U}_{2}) \leq \nu(\mathcal{U}_{1})+ \nu(\mathcal{U}_{2})$.
\begin{center}
$\nu(\mathcal{U}_{1} \cup \mathcal{U}_{2})=\underset{u \in \mathcal{U}_{1} \cup \mathcal{U}_{2}}{\sum} \nu (u)$\\
\vspace{0.2 cm}
$\Rightarrow \nu(\mathcal{U}_{1} \cup \mathcal{U}_{2})=\underset{u_{i} \in \mathcal{U}_{1}}{\sum} \nu (u_{i}) + \underset{u_{j} \in \mathcal{U}_{2}}{\sum} \nu (u_{j})- \underset{u_{x} \in \mathcal{U}_{1} \cap  \mathcal{U}_{2}}{\sum} \nu (u_{x})$\\
\vspace{0.2 cm}
$\Rightarrow \nu(\mathcal{U}_{1} \cup \mathcal{U}_{2})= \nu(\mathcal{U}_{1}) + \nu(\mathcal{U}_{2}) - \underset{u_{x} \in \mathcal{U}_{1} \cap  \mathcal{U}_{2}}{\sum} \nu (u_{x})$\\
\vspace{0.2 cm}
$\Rightarrow \nu(\mathcal{U}_{1} \cup \mathcal{U}_{2}) \leq \nu(\mathcal{U}_{1}) + \nu(\mathcal{U}_{2})$\\
\end{center}
Hence, it is proved that the BIM Game is sub\mbox{-}additive.
\end{proof}
\subsection{Overview of the Proposed Methodologies} \label{SubSec:Overview}
Here, we describe the shapley value\mbox{-}based iterative approach for identifying the influential users from a social network for the BIM Problem. The proposed methodology is broadly divided into two steps: (i) Shapley Value Computation, and (ii) Seed Set Selection. Now, we explain both these steps in detail.
\paragraph{\underline{\textbf{Step 1 (Shapley Value Computation):}}}
In this step, the Shapley value for all the nodes of the network are computed. As mentioned previously, the Shapley value of a node is basically, the average of the marginal contributions over all possible grand coalition. Now, for forming the grand coalition players may arrive at any order. Certainly, there are $n!$ possible ways to form the grand coalition. Hence, by starling approximation, the growth of $n!$ is basically of $\mathcal{O}(n^{n})$; i.e.; exponential \cite{cormen2009introduction}. Hence, even for a small size network (let's say, comprising of $100$ nodes, though real\mbox{-}world networks are much much larger) exact computation of shapley value is not feasible. Hence, to get rid of this problem, we use a result from an existing study by \cite{maleki2013bounding}. Their study shows that if the range of the marginal contribution of the players are known, it is possible to compute the shapley value of the players with bounded error and high probability. Particularly, if the number of permutations (denoted by $\tau$) considered is greater than or equal to $ \ceil*{\frac{ln(\frac{2}{\delta}). r^{2}}{2. \epsilon^{2}}}$, then the probability that the incurred error in shapley value computation for any of the players is bounded by $\epsilon$ with probability $1-\delta$. By this principle, for a given $\epsilon$, $\delta$, and $r$ value, we can easily calculate the number of permutations required to consider for shapley value computation such that all the conditions are met. Algorithm \ref{Algo:Step1} describes this procedure.
\paragraph{\underline{\textbf{Step 2 (Seed Set Selection):}}} Based on the computed shapley value, we select the seed nodes using the following ways:
\begin{itemize}
\item \textbf{Method 1:} Once the shapley value of the players are computed, we sort them based on these values and iteratively pick up users as seed nodes until the budget is exhausted. However, in each iteration once a node has been picked up as a seed node, if un\mbox{-}flag its neighbors so that these can not be selected as seed. This will help the proposed methodology to uniformly spread the seed nodes across the network. Algorithm \ref{Algo:Step21} describes this procedure.

\item \textbf{Method 2:} In this method, we detect the inherent community structure of the network. For this purpose the \emph{Louvian Algorithm} \cite{blondel2008fast} has been used in our study. After that the total allocated budget has been divided among the communities proportional to the shapley value of the nodes. Based on this shared budget, from each of the communities high shapley value nodes are chosen  as seed nodes until the budget is exhausted. If there are any extra budget during seed set from the communities, then it is transferred to the largest community. Algorithm \ref{Algo:Step22} describes this procedure.
\end{itemize}
Next, we describe both the methodologies in detail in the next subsection.
\subsection{Algorithms in the Proposed Methodology} \label{SubSec:Algo}
Now, we describe the proposed approaches in the form of algorithms.
Algorithm \ref{Algo:Step1} reports the Step 1 of our proposed methodology. The working principle of this algorithm is as follows: Given $\epsilon$, $\delta$, and $r$, first we generate the number of permutations required for shapley value computation. Next, for each permutation we compute the marginal gain of the nodes. In Algorithm \ref{Algo:Step1}, $S_{u}(p_i)$ denotes the set of nodes that appeared before $u$ in the permutation $p_i$. First, we activate the first node in the permutation and compute the influence. Next, we consider the second node in the permutation. If it is  already activated by the first node then its marginal contribution is $0$. If it is not activated then we compute the marginal gain by subtracting the individual influence from the combined influence. In this way, we compute the marginal gain of all the nodes of the network. We repeat this process for $\mathcal{R}$ (part of input) given number of times. Now, we compute the marginal gain by dividing $\mathcal{R}$. Finally, the shapley value is computed by dividing the number of permutations considered in shapley value computation.
\par It is important to observe that Algorithm \ref{Algo:Step1} takes the range of marginal gain of the players as one of the input among many. 
Now, we mention the way we compute this quantity. For any node  $u \in V(G)$, we compute its range as follows: Consider the neighbors of the node, i.e.,  $N(u)$. Consider a node $w \in N(u)$. Now, $w$ can be influenced by any of its neighbors. The upper value of this range is $(\frac{deg(u)}{2}+\underset{w \in N(u)}{\sum} \frac{1}{deg(w)-1})$. In the worst case may be none of the neighbors will be influenced, and hence, the lower value of the range is $0$. So, the range of the marginal gain of the player $u$ is $[0,\frac{deg(u)}{2}+\underset{w \in N(u)}{\sum} \frac{1}{deg(w)-1}]$. After computing the range of all the players, we aggregate them by taking average over all the players; i.e., $r=\frac{\underset{u \in V(G)}{\sum c_{u}}}{n}$.
\par As mentioned, after computing the shapley value of the players the next step is to choose the seed nodes. Algorithm \ref{Algo:Step21} describes the Method 1. Though Algorithm \ref{Algo:Step21} easy to understand and simple to implement, still we can improve the performance of this algorithm by exploiting the community structure of the network. Algorithm \ref{Algo:Step22} describes that procedure.
\pagebreak

  \begin{algorithm}[H]
	\SetAlgoLined
	\KwData{A BIM Problem Instance $<G(V, E, P),C,B>$, Error $\epsilon$, Probability $\delta$, and Range of Marginal Contribution $r$.}
	\KwResult{The array $\phi$ containing the shapley value of the nodes.}
	\For{$u \in V(G)$}{
	$b_{u} \longleftarrow 0$\;
	$c_{u} \longleftarrow \frac{deg(u)}{2}+\underset{w \in N(u)}{\sum} \frac{1}{deg(w)-1}$\;
	}
	$r= \frac{\underset{u \in V(G)}{\sum c_{u}}}{n}$\;
	Generate $\tau \geq \ceil*{\frac{ln(\frac{2}{\delta}). r^{2}}{2. \epsilon^{2}}}$ number of permutations of the players\;
	Computation repeats for $\mathcal{R}$ times \;
	Create the arrays $\phi$ and $MG$, each of size $n$ and initialize them with $0$\;
	\For{ $i=1 \text{ to } \tau$}{
	\For{ $u \in V(G)$}{
	$T(u)=0$\;
	}
	\For{ $j=1 \text{ to } \mathcal{R}$}{
	\For{ $u \in V(G)$}{
	$T(u)=T(u)+ \nu(S_{u}(p_i) \cup \{u\}) - \nu(S_{u}(p_i))$\;
	}
	$MG(u)= \frac{T(u)}{\mathcal{R}}$\;
	}
	}
	\For{ $u \in V(G)$}{
	$\phi(u)= \frac{MG(u)}{\tau}$\;
	}
	Return the array $\phi$ containing the shapley value of the nodes
	\caption{Step 1: Algorithm for Computing Approximate Shapley Value Computation}
	\label{Algo:Step1}
\end{algorithm}

 \begin{algorithm}[h]
	\SetAlgoLined
	\KwData{A BIM Problem Instance $<G(V, E, P),C,B>$, The array $\phi$.}
	\KwResult{The seed set for diffusion $S$, such that $C(S) \leq B$.}
	$N \longleftarrow$ Sort the players based on the shapley value computed in $\phi$. \;
	$S= \emptyset$\;
	Declare the array \texttt{Flag} of size $n$ and initialize each entry by $1$ \;
	\For{$i=1 \text{ to }n$}{
	\While{$B>0$}{
	 \If{$B \geq C(u_{i}) \text{ and } Flag(u_{i})==1$}{
	 $S \longleftarrow S \cup \{u_i\}$\;
	 $B \longleftarrow B - C(u_{i})$\;
	 }
	 \For{$\text{All } v \in \mathcal{N}(u_i)$}{
	 $Flag(v) \longleftarrow 0$\;
	 }
	}
	}
	Return the set $S$ as seed set.
	\caption{Step 2: Influential Node Selection (Method $1$)}
	\label{Algo:Step21}
\end{algorithm}

\begin{algorithm}[H]
	\SetAlgoLined
	\KwData{A BIM Problem Instance $<G(V, E, P),C,B>$, The array $\phi$.}
	\KwResult{The seed set for diffusion $S$, such that $C(S) \leq B$.}
	$Community \longleftarrow \text{Detect the Community Structure of the Network} $\;
	$K \longleftarrow \{K_1, K_2, \ldots, K_\ell\}$;
	$K \longleftarrow Sort(K)$\;
	$K_{max} \longleftarrow Largest\_Community(K)$;
	$Total\_Shapley \longleftarrow 0$\;
	\For{$\text{All } u \in N$}{
	$Total\_Shapley= Total\_Shapley+\phi (u)$\;
	}
	Create the array $B_{K}$ and initialize all entries with $0$\;
	\For{$i=1 \text{ to } \ell$}{
	$Temp\_Shapley \longleftarrow 0$\;
	\For{$\text{All } u \in V(K_{i})$}{
	$Temp\_Shapley = Temp\_Shapley + \phi(u)$\;
	}
	$B_{K_{i}} \longleftarrow \frac{Temp\_Shapley}{Total\_Shapley} .B$
	}
	$S \longleftarrow \emptyset$\;
	\For{$i=1 \text{ to } \ell$}{
	$V(K_{i}) \longleftarrow$ Sort the nodes in $V(K_{i})$ based on the value computed in $\phi$\;
	\For{$j=1 \text{ to } |V(K_{i})|$}{
	\While{$B_{K_{i}} \geq 0$}{
	\If{$B_{K_{i}} \geq C(u_j)$}{
	$S \longleftarrow S \cup \{u_j\}$\;
	$B_{K_{i}} \longleftarrow B_{K_{i}} - C(u_j)$\;
	}
	}
	}
	$B_{K_{max}} \longleftarrow B_{K_{max}}+ B_{K_{i}}$\;
	}
	Return the set $S$ as seed set.
	\caption{Step 2: Influential Node Selection (Method $2$)}
	\label{Algo:Step22}
\end{algorithm}

\subsection{Complexity Analysis of the Proposed Methodologies} \label{SubSec:Complexity}
Now, we analyze the algorithms to understand the running time and space requirement for the proposed methodologies. From Line $1$ to $5$, we compute the range of marginal gain of the players. This can be implemented in $\mathcal{O}(n^{2})$ time and $\mathcal{O}(n)$ space. Degrees of all the nodes can be computed in $\mathcal{O}(n^{2})$ time. Now, assume that $d_{max}$ denotes the maximum degree of all the nodes; i.e.; $d_{max}=\underset{u \in V(G)}{max} \ deg(u)$. Hence, for a particular player $u \in V(G)$, the instruction mentioned in Line 3 can be computed in $\mathcal{O}(d_{max})$ time. The execution time of Line No $5$ is of $\mathcal{O}(n)$. Hence, the computational time from Line $1$ to $5$ is of $\mathcal{O}(n^{2}+n d_{max}+n)$. As, $d_{max}$ is upper bounded by $\mathcal{O}(n)$, hence the quantity reduces to $\mathcal{O}(n^{2})$. It is important to observe that marginal contribution of any player can be computed by traversing the graph, which takes $\mathcal{O}(m+n)$. In each permutation, there are $n$ players. Hence, considering $\tau$ number of permutations and repeating the same experiment for $\mathcal{R}$ times  requirement for marginal gain computation is of $\mathcal{O}(\tau \cdot n \cdot (m+n) \cdot \mathcal{R})$ time. After that, computing the shapley value requires additional $\mathcal{O}(n)$ time. Hence, total time requirement is of $\mathcal{O}(n^{2} + \tau \cdot n \cdot (m+n) \cdot \mathcal{R})$ and this reduces to $\mathcal{O}(\tau \cdot n \cdot (m+n) \cdot \mathcal{R})$. The extra space consumed by the Algorithm \ref{Algo:Step1} is to store the arrays $T$, $MG$, $\phi$, also to store the degree of the nodes and all of them requires $\mathcal{O}(n)$ space each. Hence the Lemma \ref{Lemma:3} holds.
\begin{mylemma} \label{Lemma:3}
Running time and space requirement of Algorithm \ref{Algo:Step1} is of $\mathcal{O}(\tau \cdot n \cdot (m+n) \cdot \mathcal{R})$ and $\mathcal{O}(n)$, respectively.
\end{mylemma} 

As mentioned previously, once the shapley value of the players are computed, the seed set can be selected either of the two ways. In Algorithm \ref{Algo:Step21}, first the nodes of the players are sorted based on the shapley value computed in Algorithm \ref{Algo:Step1}. This step requires $\mathcal{O}(n \log n)$ time. Let, $d_{max}$ denotes the maximum degree of the network. Now, from Line $4$ to $14$ requires $\mathcal{O}(n d_{max})$ time. Hence, the running time of Algorithm \ref{Algo:Step21} is of $\mathcal{O}(n \log n + n d_{max})$ time. Extra space required by Algorithm \ref{Algo:Step21} is to store the array $Flag$ and the seed set $S$, which is of $\mathcal{O}(n+|S|)=\mathcal{O}(n)$. Hence, Lemma \ref{Lemma:4} holds.
\begin{mylemma} \label{Lemma:4}
The running time and space requirement of Algorithm \ref{Algo:Step21} is of $\mathcal{O}(n \log n + n d_{max})$, and  $\mathcal{O}(n)$, respectively.
\end{mylemma}
Algorithm \ref{Algo:Step22} describes another way of selecting seed nodes. In this method, first the community structure is detected. This can be done by the Louvian Algorithm, which takes $\mathcal{O}(n \log n)$ time \footnote{\url{https://en.wikipedia.org/wiki/Louvain_modularity}}. The array $Community$ stores the community number corresponding to each user, i.e., $Community[i]=x$ means the user $u_i$ belongs to the $x$\mbox{-}th community of the network. Assume that there  are $\ell$ number of communities of the network. Sorting the communities will require $\mathcal{O}(\ell \log \ell)$ time. Computing the total shapley value of the network requires $\mathcal{O}(n)$ time. Line $10$ to $16$ actually shows the budget distribution among the communities, which takes $\mathcal{O}(n \ell)$ time. Line $18$ to $29$ describes the seed set selection process. Now, it is important to observe that the running time of this phase depends upon the number of nodes that the community contains. To give a weak upper bound, we first calculate the running time for the seed selection for the largest community, and then multiply it with the number of communities. Let, $n_{max}$ denotes the number of nodes in the largest community. $\mathcal{O}(n)$ time is required to identify the nodes belongs to that community and sorting them based on the values computed in the array $\phi$ requires $\mathcal{O}(n_{max} \log n_{max})$ time. Now, selecting the nodes from the sorted list requires $\mathcal{O}(n_{max})$ time. So, the time requirement for the execution of the largest is of $\mathcal{O}(n + n_{max} \log n_{max} + n_{max})= \mathcal{O}(n + n_{max} \log n_{max})$. Additionally, during the processing of communities other than the largest one, transferring the unutilized budget of the community to the largest community requires $\mathcal{O}(1)$ time. Hence, time requirement for the execution from Line $18$ to $29$ requires $\mathcal{O}(n \ell + \ell n_{max} \log n_{max}+\ell)=\mathcal{O}(n \ell + \ell n_{max} \log n_{max})$. Hence, total time requirement of Algorithm \ref{Algo:Step22} requires $\mathcal{O}(n \log n + \ell \log \ell + n \ell+n \ell + \ell n_{max} \log n_{max})= \mathcal{O}(n \log n + \ell \log \ell + n \ell + \ell n_{max} \log n_{max})$. Extra space consumed by Algorithm \ref{Algo:Step22} is to store the array $Community$, which requires $\mathcal{O}(n)$ space; the array $B_{K}$, which requires $\mathcal{O}(\ell)$ space, vertices of the communities during seed set selection, which requires $\mathcal{O}(n_{max})$ space, and storing the seed set which requires $\mathcal{O}(|S|)$ space. Hence, total space requirement of Algorithm \ref{Algo:Step22} is of $\mathcal{O}(n + \ell +n_{max}+ |S|)= \mathcal{O}(n + \ell)$. Hence, Lemma \ref{Lemma:5} holds.
\begin{mylemma} \label{Lemma:5}
Running time and space requirement of Algorithm \ref{Algo:Step22} is of $\mathcal{O}(n \log n + \ell \log \ell + n \ell + \ell n_{max} \log n_{max})$ and $\mathcal{O}(n + \ell)$, respectively.
\end{mylemma}

Now, after computing the shapley value by Algorithm \ref{Algo:Step1}, if we use Algorithm \ref{Algo:Step21} for selecting the seed nodes then the running time and space requirement of the proposed methodology becomes is of $\mathcal{O}(\tau(m+n)R + n \log n + n d_{max})$, and $\mathcal{O}(n)$, respectively. Otherwise, if we use Algorithm \ref{Algo:Step22} for seed set selection then the running time and space requirement of the proposed methodology becomes $\mathcal{O}(\tau \cdot n \cdot (m+n) \cdot \mathcal{R} + n \log n + \ell \log \ell + n \ell + \ell n_{max} \log n_{max})$, and $\mathcal{O}(n + \ell)$, respectively. Hence, from Lemma \ref{Lemma:3}, \ref{Lemma:4}, and \ref{Lemma:5}, following  two theorems are implied.
\begin{mytheorem}
If Algorithm \ref{Algo:Step21} is used for seed set selection then the running time and space requirement of our proposed methodology becomes $\mathcal{O}(\tau \cdot n \cdot (m+n) \cdot \mathcal{R} + n \log n + n d_{max})$, and $\mathcal{O}(n)$, respectively. 
\end{mytheorem}

\begin{mytheorem}
If Algorithm \ref{Algo:Step22} is used for seed set selection then the running time and space requirement of our proposed methodology becomes $\mathcal{O}(\tau \cdot n \cdot (m+n) \cdot \mathcal{R} + n \log n + \ell \log \ell + n \ell + \ell n_{max} \log n_{max})$, and $\mathcal{O}(n + \ell)$, respectively.
\end{mytheorem}
\section{Experimental Evaluation} \label{Sec:EE}
In this section, we describe the experimental evaluation of the proposed methodologies. Also, the obtained results have been compared with that of results obtained from the existing methods from the literature. Initially, we start with a brief description of the datasets used in our experiments. 
\subsection{Datasets} We use the following three datasets appeared in two  different situations. First two datasets are basically collaboration networks among the researchers of two different areas. The third one is the trust network among the users of the E-Commerce house \texttt{epinions.com}. 
\begin{itemize}
\item \textbf{HEP Theory Collaboration Network (HEPT)} \cite{leskovec2007graph} \footnote{\url{https://arxiv.org/archive/hep-th .}}: This dataset contains the collaboration information among the high energy physics researchers obtained from the papers submitted in the high energy physics section of \url{arxiv.org}. Here, individual researchers are the nodes and two nodes are linked by an edge if the corresponding researchers co\mbox{-}authored atleast one paper. 
\item \textbf{Condensed Matter Physics Collaboration Network (CMP)} \footnote{\url{https://arxiv.org/archive/cond-mat}} \cite{leskovec2007graph}: This is also a collaboration network and obtained by connecting the researchers of the condensed matter physics section of the \url{arxiv.org}.
\item \textbf{Epinions Dataset} \footnote{\url{http://www.epinions.com/ }} \cite{richardson2003trust}: This dataset contains `who trust whom' information among the users of a review site named `Epinions'. Here, the users form the vertex set of the network and there is a directed edge between $u_i$ and $u_j$ if and only if $u_i$ trusts $u_j$.
\end{itemize}
All these datasets have been previously used in the experimentation in the domain of influence maximization \cite{jung2012irie, tong2016adaptive,wen2020identification}. Table \ref{Tab:Data_Stat} contains basic statistics of the mentioned datasets.
\begin{table}[H]
\centering
\caption{Basis Statics of the Datasets}
\label{Tab:Data_Stat}
 \begin{tabular}{||c c c c c c c||} 
 \hline
 Dataset Name & Type & $n$ & $m$ & $d_{avg}$ & $|K|$ & Modularity \\ [0.5ex] 
 \hline\hline
 HEPT & Undirected & 9877 & 25998 & 5.26 & 481 & 0.76382\\ 
 CMP & Undirected &23133 & 93497 & 8.08 & 20 & 0.68561\\
 Epinions & Directed & 75879 & 508837 & 13.41 & 296 & 0.79866\\ [1ex] 
 \hline
 \end{tabular}
\end{table}
\subsection{Experimental Setup}
Now, we state the experimental set up that has been used in this paper. Following parameters are there in our study whose value need to be set:
\begin{itemize}
\item \textbf{Selection Cost}: We assign an integer values as the selection cost of the users from the interval $[50,100]$ uniformly at random. In \cite{nguyen2013budgeted}'s study also cost of the users are chosen randomly from a fixed interval. 
\item \textbf{Budget}: In our experiments, we start with a budget value of $2000$ and each time incremented by $4000$ and continued till $26000$. Hence, we experiment with following budget values $B=\{2k, 6k,10k, 14k, 18k, 22k, 26k\}$. In \cite{nguyen2013budgeted}'s study also experimentation is carried out with some fixed budget values.
\item \textbf{Influence Probability}: As per the existing studies in the literature, in this study also we consider the following three influence probability settings:
\begin{itemize}
\item \textbf{Uniform}: In this setting, every edge of the network is associated with the same probability value. In this study, we consider this fixed value as $0.1$ and $0.2$. These two value has been considered by many existing studies. 
\item \textbf{Tri\mbox{-}Valency}: In this setting all the edges of the network are assigned with a fixed probability from the set $\{0.1, 0.01,0.001\}$ uniformly at random. This setting has also been considered  in many existing studies in the literature.
\item \textbf{Weighted Cascade}: In this setting every edge $(uv)$ has the influential probability which is equal to the reciprocal of the degree of $v$; i.e.; $P_{uv}=\frac{1}{deg(v)}$. In case of directed network $deg(v)$ will be replaced by $indeg(v)$. This setting is also common in many existing studies.
\end{itemize}
All these influence probability settings have been used in existing studies in and around the influence maximization problem \cite{hong2019seeds,chen2010scalable,logins2019content,cohen2014sketch}.
\item \textbf{Range of the marginal gain of each player}: As mentioned previously, to compute the shapley value of a node, we need to know its range of the marginal gain of the players. For any node  $u \in V(G)$, we compute its range as follows: Consider the neighbors of the node, i.e.,  $N(u)$. Consider a node $w \in N(u)$. Now, $w$ can be influenced by any of its neighbors. The upper value of this range is $(\frac{deg(u)}{2}+\underset{w \in N(u)}{\sum} \frac{1}{deg(w)-1})$. In the worst case may be none of the neighbors will be influenced, and hence, the lower value of the range is $0$. So, the range of the marginal gain of the player $u$ is $[0,\frac{deg(u)}{2}+\underset{w \in N(u)}{\sum} \frac{1}{deg(w)-1}]$. 
\end{itemize}
\subsection{Goals of the Experiment}
Here, we mention the goals of the experiments, which are mentioned below:
\begin{itemize}
\item The primary goal of our experimentation is to make a comparative study among the proposed as well as existing methodologies regarding the  quality of the seed set that they can select.
\item Our second concern is the computational time. We also make a comparative study of the proposed as well as baseline methods based on their computational time requirement.  
\end{itemize}
\subsection{Algorithms in the Experiment}
Now, we mention the algorithms that are there in our experiments: 
\subsubsection{Algorithms Proposed in this Paper}
\begin{itemize}
\item \textbf{BIM with Co-Operative Game Theory (BIMGT)}: In this method, the marginal contribution in `shapley value' for every user of the network is computed and the users are sorted based o this value. Users are chosen from this sorted list until the budget is exhausted. This is basically `Method 1' of this paper.
\item \textbf{BIM with Co-Operative Game Theory and Community Structure (BIMGTC)}: In this method, first the community structure of the network has been exploited and based on the total shapley value of the nodes of the community the total budget is divided among the communities. Subsequently, the nodes with high shapley value is chosen as seed nodes from each of the communities until the allocated budget from each of the communities are exhausted. 
\end{itemize}
\subsubsection{Algorithms from the Literature}
We compare the performance of the proposed solution approaches with the following existing methods from the literature:
\begin{itemize}
\item \textbf{Random (RAND)}: In this method, starting with an empty seed set in each iteration among the non-seed nodes any arbitrary node is chosen uniformly at random and put it into the seed set. This process is repeated until the allocated budget is exhausted. Many existing studies considers this method as a baseline \cite{kempe2003maximizing,wu2014coritivity}.
\item \textbf{Maximum Degree Heuristic (MDH)}: In this method the nodes are sorted based on its degree. Next, the high degree nodes are chosen as the seed set until the budget is exhausted. This method has been used in many existing studies on influence maximization \cite{narayanam2010shapley,wu2014coritivity,shang2017cofim}.
\item \textbf{Maximum Clustering Co-efficient Heuristic (MCCH)}: Working principle of this method is same as the MDH, though instead of degree, in this method the clustering co-efficient of the nodes are computed. Nodes are sorted based on this value and the nodes with high clustering efficient values are chosen as seed nodes unless the allocated budget is exhausted. Existing studies on influence maximization have used this method as a baseline \cite{narayanam2010shapley}.
\item \textbf{DAG Based Heuristic for BIM Problem (DAGHEU)} \cite{nguyen2013budgeted}: This is the first study on the BIM Problem that contains a number of solution methodologies. Among them we compare the performance of our proposed methodologies with DAG1-SPBP. This is the most efficient and effective as per their claim.  
\item \textbf{Balanced Seed Selection Heuristic (BSSH)} \cite{han2014balanced}: In this method, first the nodes are divided into two groups: one is `cost effective' nodes and `influential' nodes. Finally, the seed nodes are chosen in an intelligent way from both the sets in an efficient way. 
\item \textbf{Community-Based Approach for the BIM Problem} \cite{banerjee2019combim}: This is the first community\mbox{-}based solution approach for the BIM Problem by \cite{banerjee2019combim}. In this method the total allocated budget is divided based on `node fraction' and `cost fraction' among the communities. Subsequently, the high degree nodes from the communities are chosen as seed nodes until the community specific budget is exhausted.
\end{itemize}
All these methodologies have been implemented in \emph{Python 3.4 + NetworkX 2.0.1} environment. Experiments of this study are performed in an workstation with Intel® Xeon(R) W-1290 CPU \@ 3.20GHz × 20  and 16 GB of RAM.
\subsection{Experimental Results with Discussions}
Now, we describe the experimental results with detailed discussion. First we discuss the impact of influence spread  for different budget values by different algorithms.
\subsubsection{Impact of Influenced Spread}
Figure \ref{Fig:Collab2} shows the budget vs. number of influenced nodes plots for the HEPT dataset under different influence probability settings. From this figure it can be observed that the seed set selected by the proposed methodologies leads to the more number of influenced nodes compared to the baseline methods. Now, we give one example. For uniform probability setting with $p_c=0.2$ and $\mathcal{B}=26000$, among the existing methodologies the seed set selected by \textbf{BIMGTC} leads to the maximum number of influenced nodes which is $3781$. On the other hand, among the existing methods the seed set chosen by \textbf{ComBIM} leads to the maximum number of influenced nodes which is $3362$. This is almost $12.5 \%$ more compared to the \textbf{ComBIM}. This observation is consistent over different influence probability settings. From Table \ref{Tab:Data_Stat}, it can be observed that the number of nodes of the HEPT dataset is $9877$. Hence, the influence coverage of the seed sets selected by \textbf{BIMGTC} and  \textbf{ComBIM} are approximately $38.28 \%$ and $34.04$. So, there is a gap of almost $3.2 \%$. It is also observed that the number of influenced nodes are least in case of trivalency model and this is irrespective of any method. As an example, for $\mathcal{B}=26000$ the seed set selected by the \textbf{BIMGTC} method leads to $1696$, $3781$, $915$, and $3112$ number of influenced nodes for uniform (with $p_c=0.1$ and $p_c=0.2$), trivalency, and weighted cascade models respectively.

\begin{figure*}[t]
\centering
\begin{tabular}{ccc}
\includegraphics[scale=0.2]{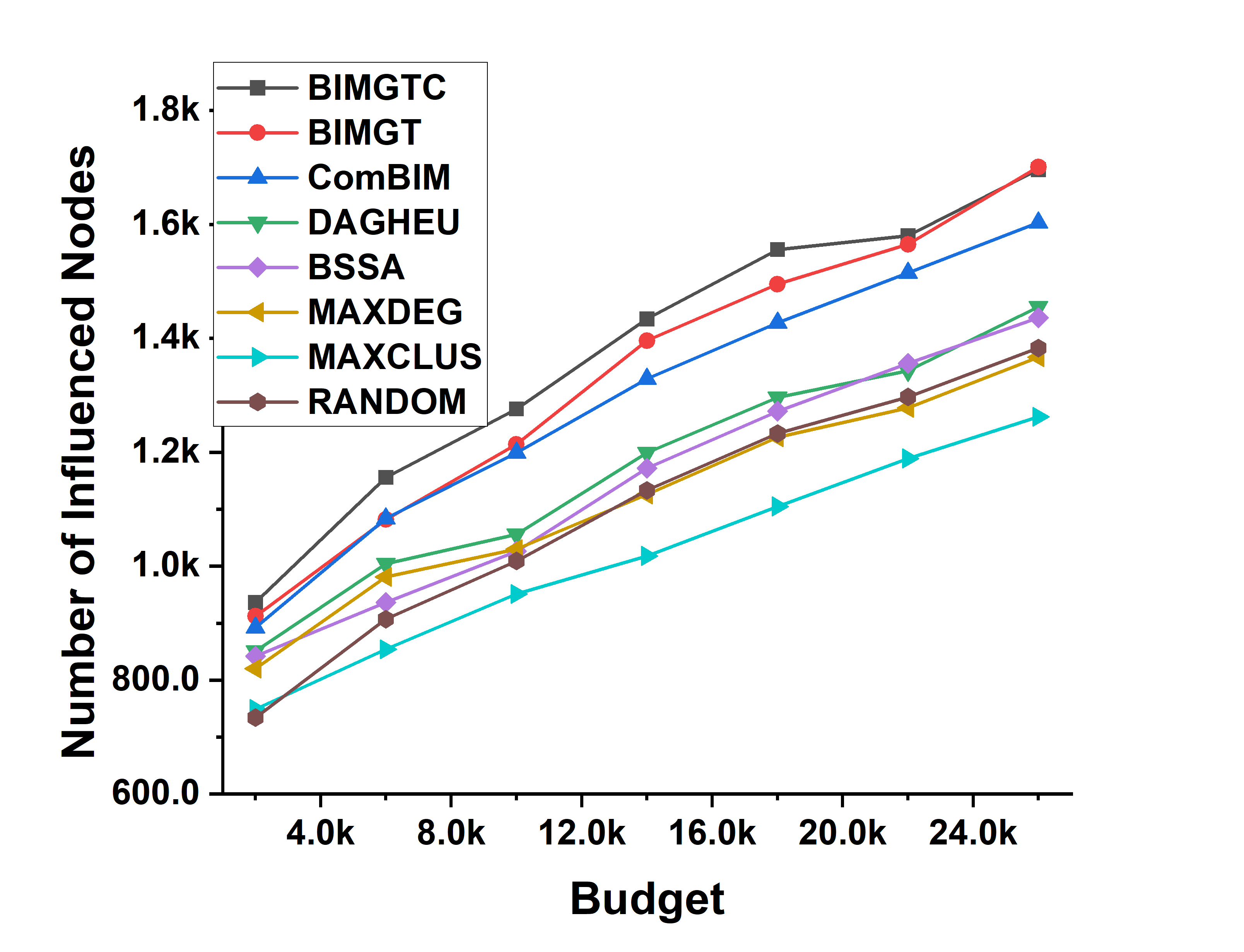} & \includegraphics[scale=0.2]{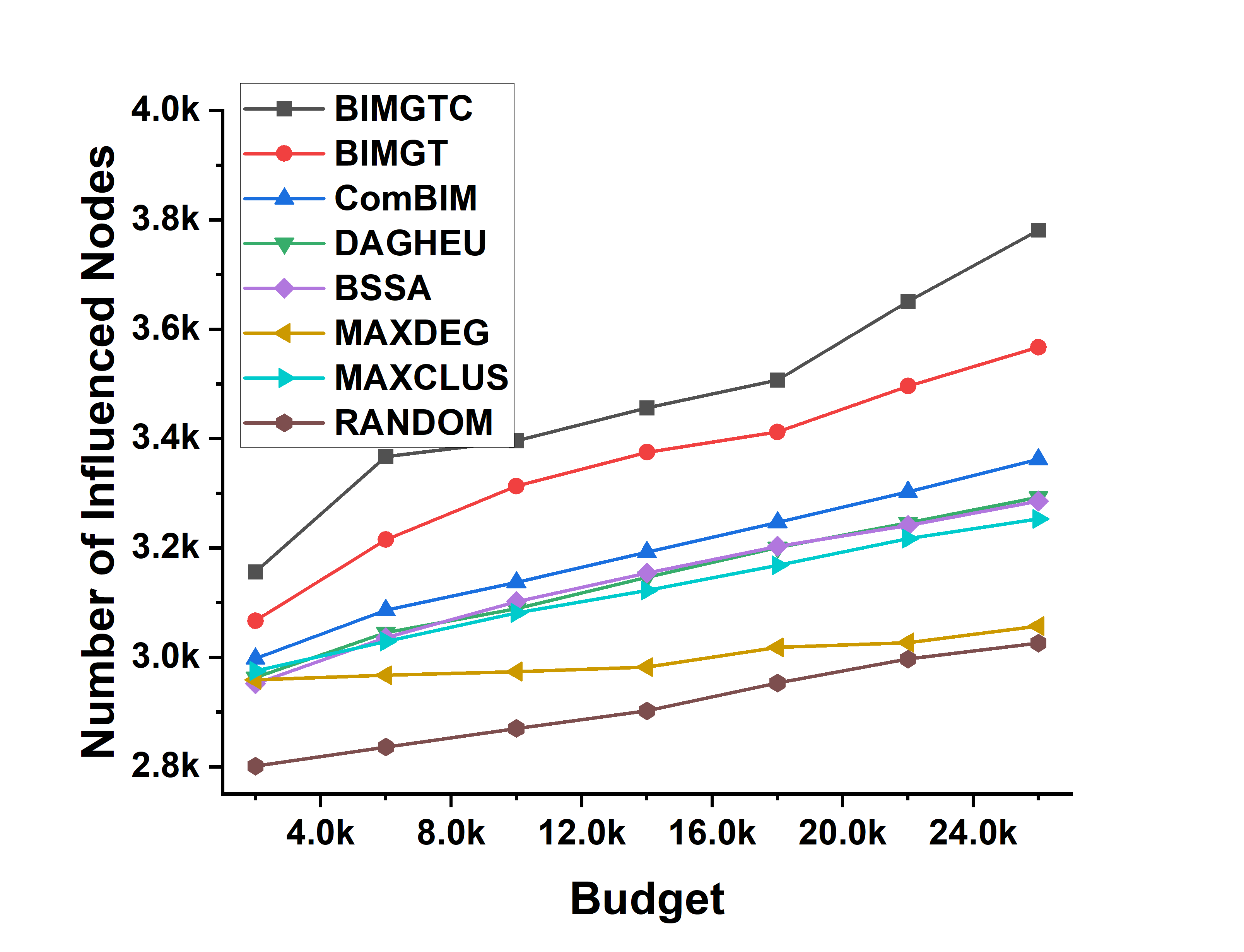} \\
(a) CMP Dataset ($p_c=0.1$)  & (b) CMP Dataset  ($p_c=0.2$) \\
\includegraphics[scale=0.2]{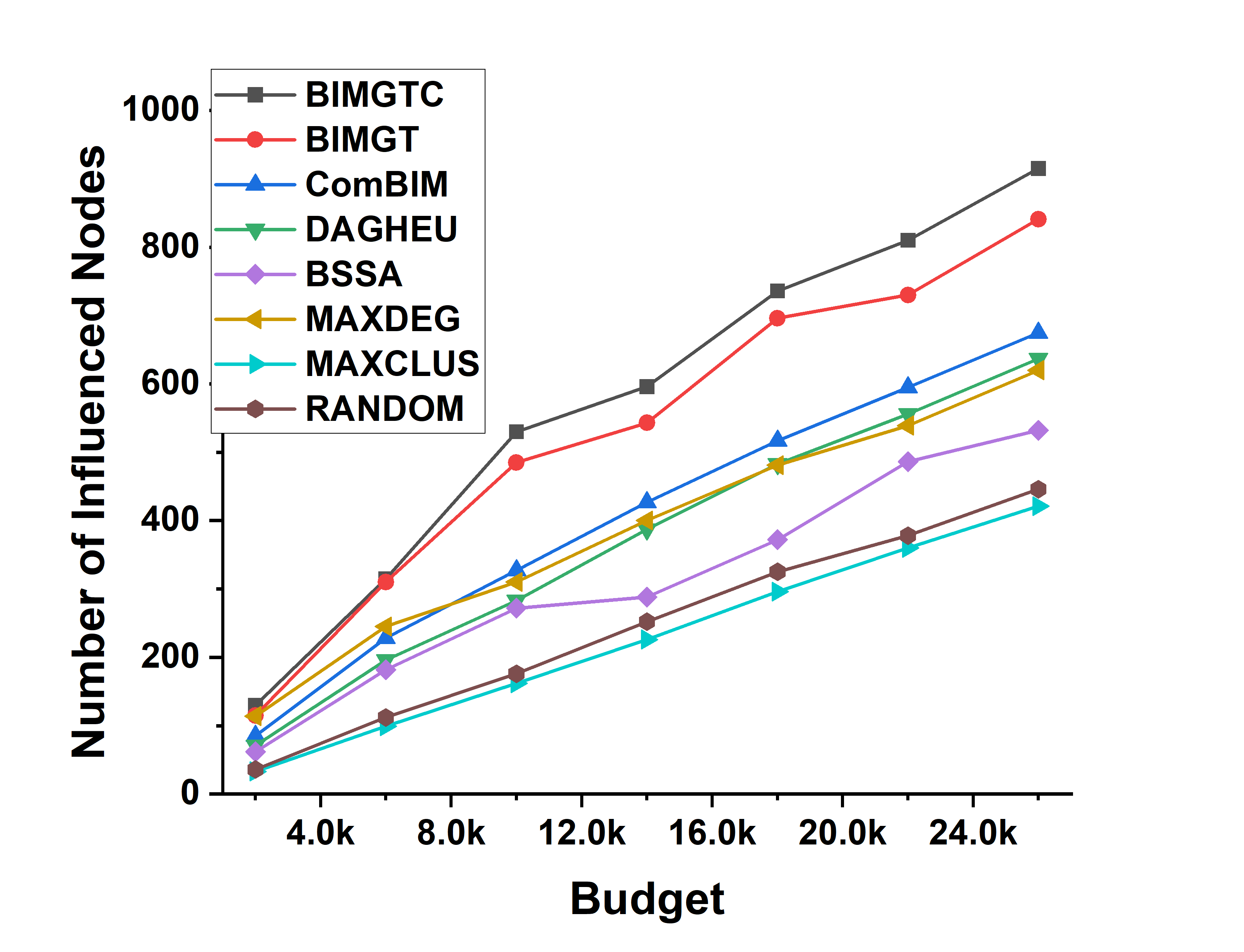} & \includegraphics[scale=0.2]{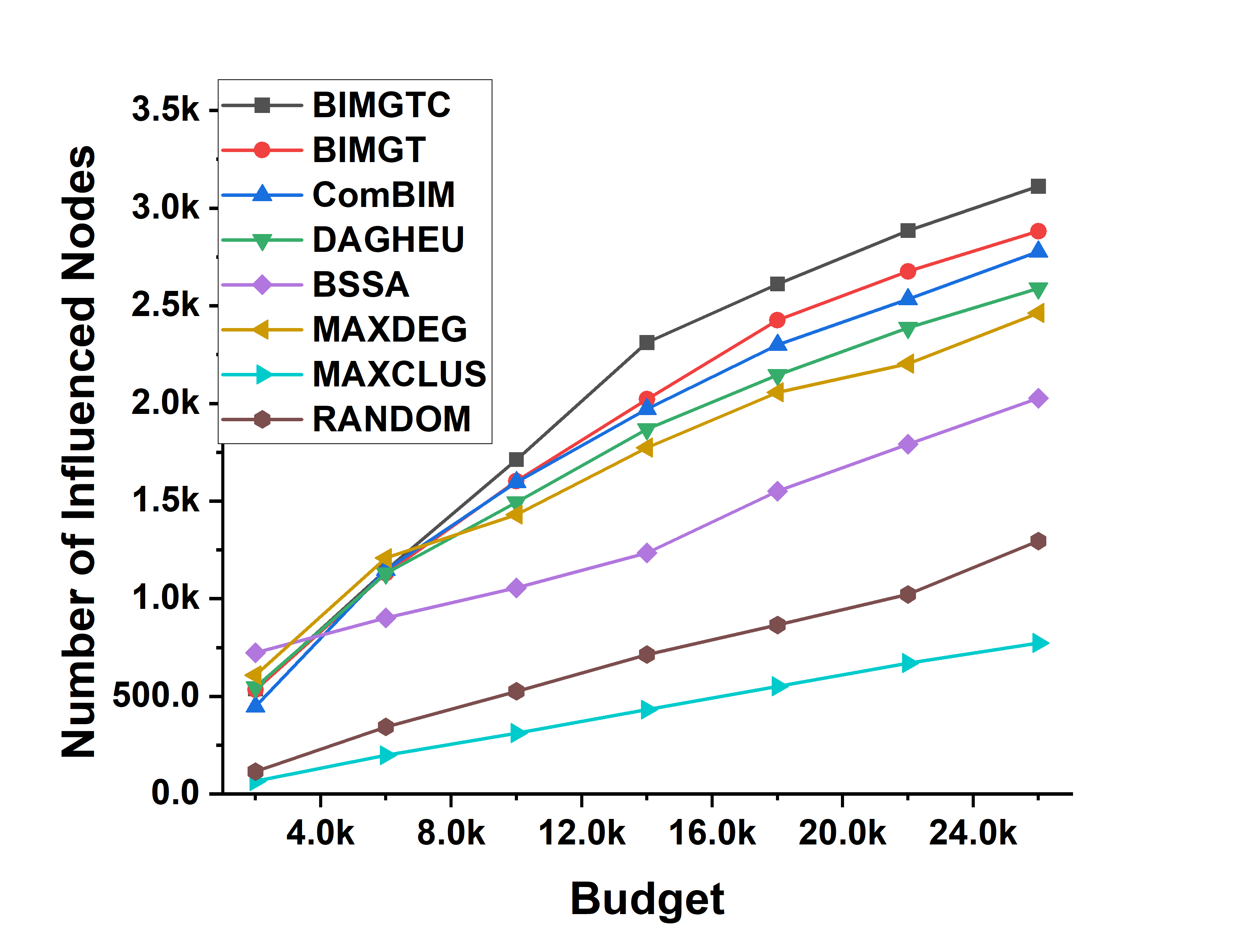} \\

(c) CMP Dataset (Trivalancy) & (d) CMP Dataset  (Weighted Cascade) \\
\end{tabular}
\caption{Budget Vs. Number of Influence Plots for the \textbf{Collab2}  dataset under the Uniform (with $p_c=0.1$ and $0.2$), Trivalency, Weighted Cascade and Count Probability settings.}
\label{Fig:Collab2}
\end{figure*}


\begin{figure*}[t]
\centering
\begin{tabular}{ccc}
\includegraphics[scale=0.2]{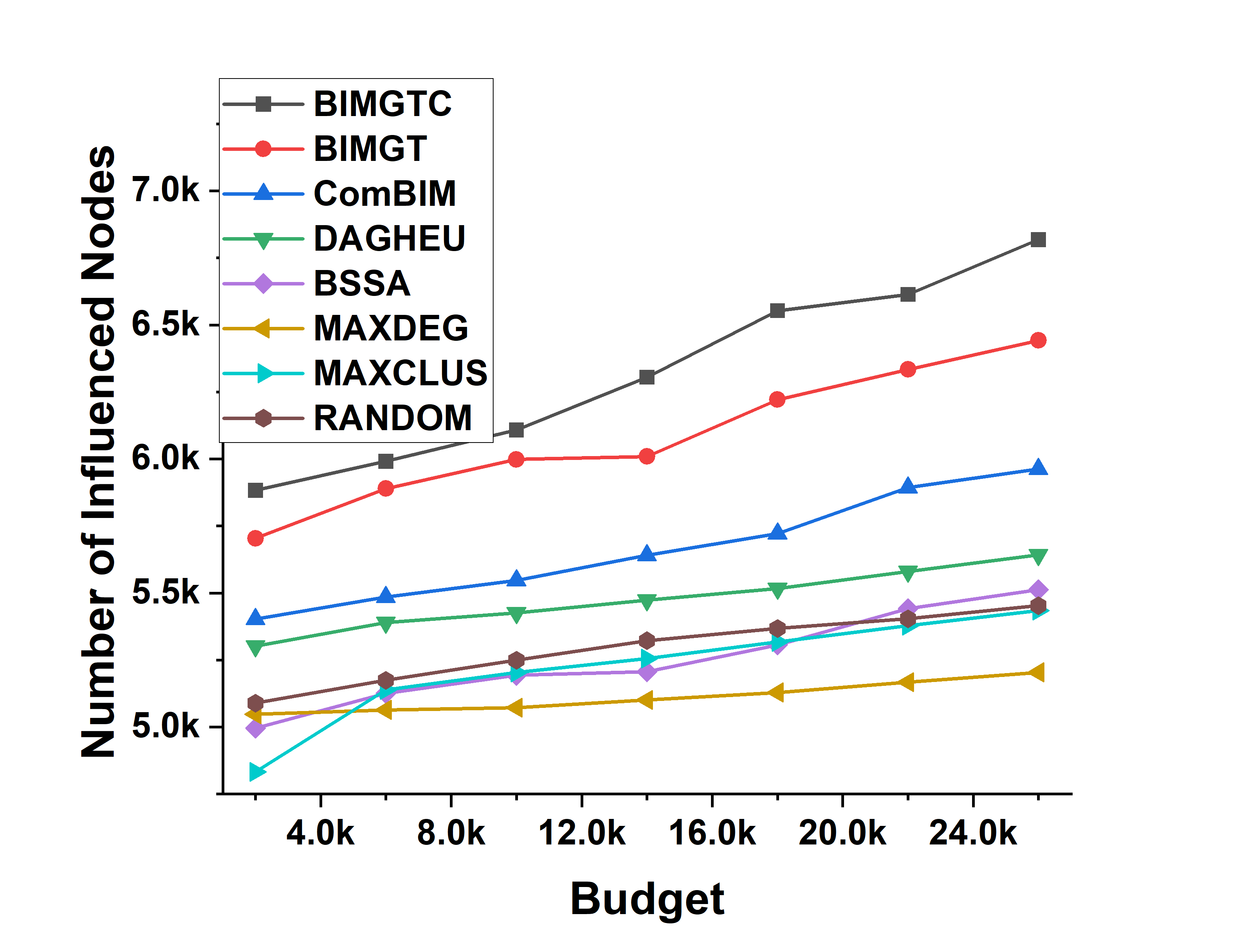} & \includegraphics[scale=0.2]{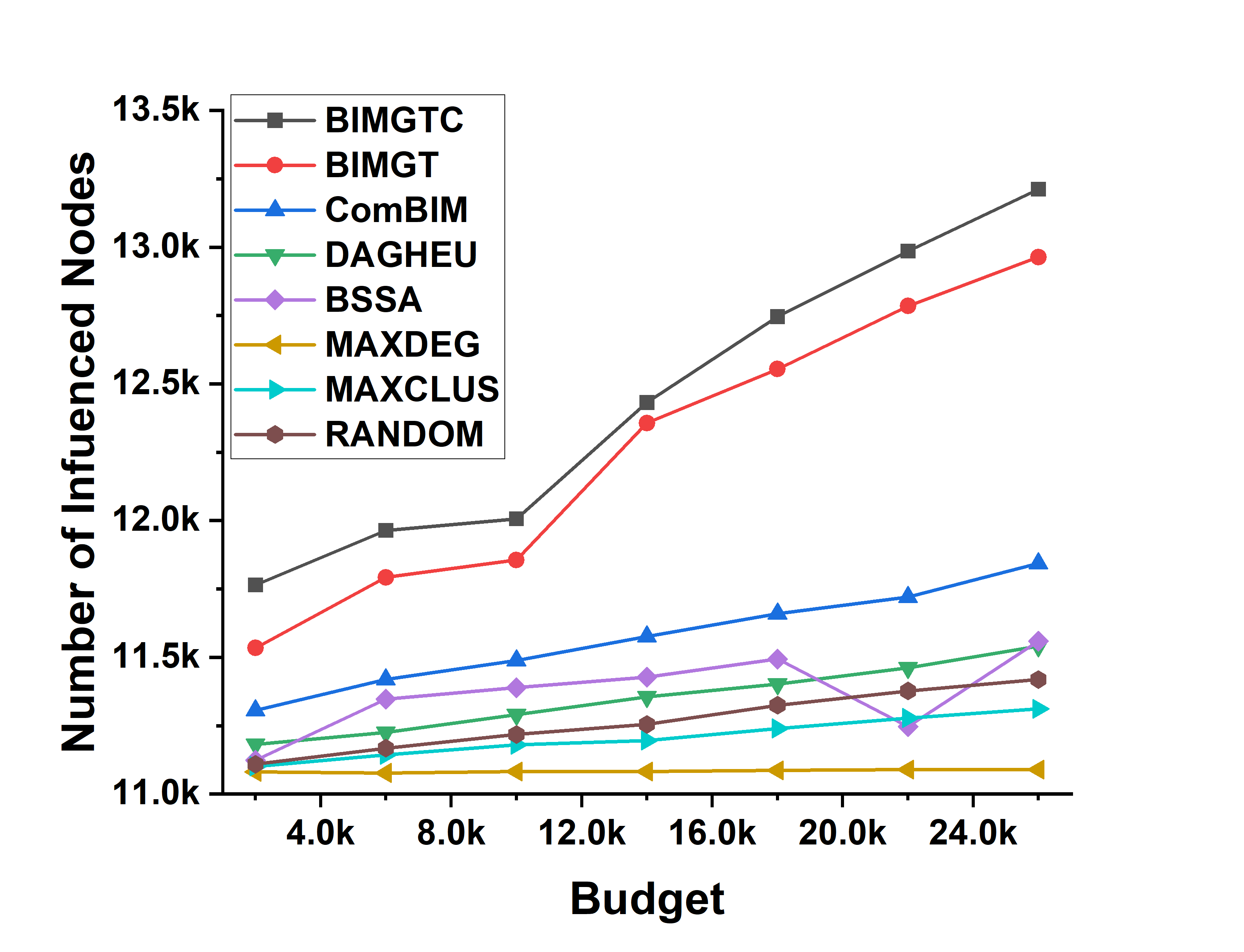} \\
(a) CMP Dataset ($p_c=0.1$)  & (b) CMP Dataset  ($p_c=0.2$) \\
\includegraphics[scale=0.2]{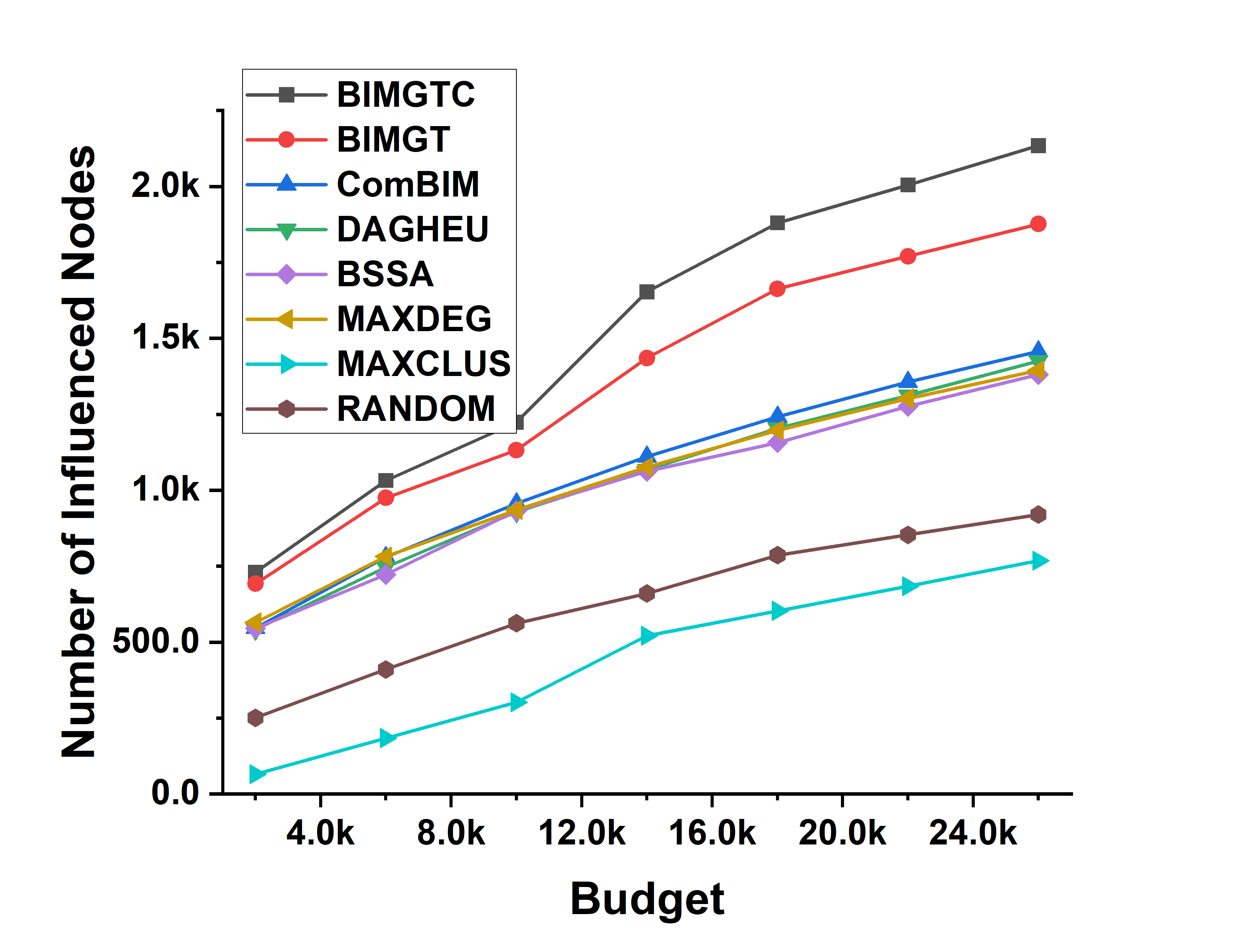} & \includegraphics[scale=0.2]{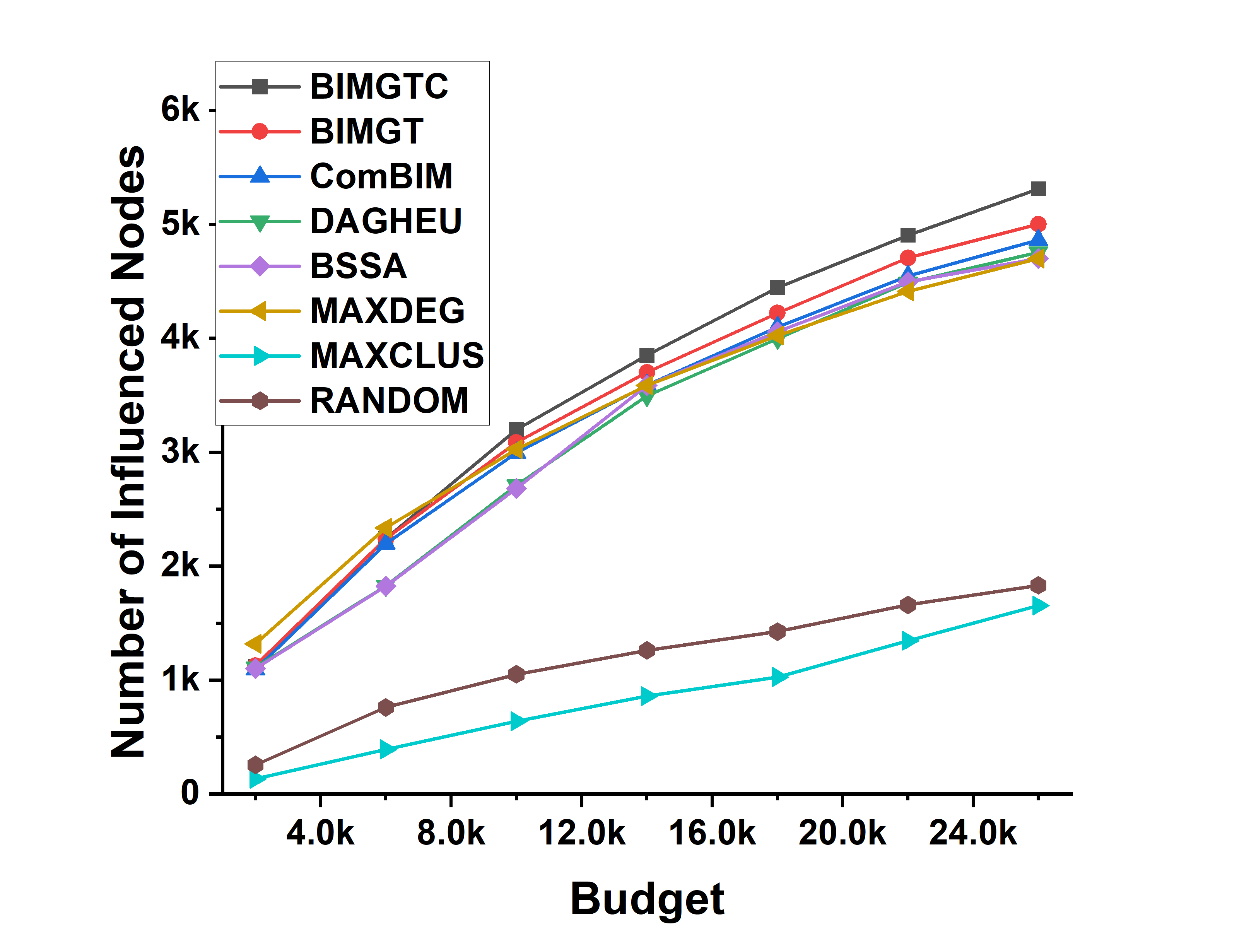} \\

(c) CMP Dataset (Trivalancy) & (d) CMP Dataset  (Weighted Cascade) \\
\end{tabular}
\caption{Budget Vs. Number of Influence Plots for the \textbf{CMP Dataset}  dataset under the Uniform (with $p_c=0.1$ and $0.2$), Trivalency, Weighted Cascade and Count Probability settings.}
\label{Fig:Collab3}
\end{figure*}

\par Next, we describe the obtained results for the CMP Dataset. Figure \ref{Fig:Collab3} shows the budget vs. number of influenced nodes plot for the CMP Dataset. Like HEPT Dataset, in this dataset also the seed set selected by the proposed methodologies leads to the more number of influenced nodes compared to the existing methods. Here, we give an example. Under uniform probability setting when all the edges of the network have the influence probability of $0.2$, the seed set selected by the \textbf{BIMGTC} method leads to more number of influence nodes which is $13212$, which is almost $57.17 \%$ of the network. Among the existing  methods, \textbf{ComBIM} leads to a seed set that results to $11844$ number of influenced nodes, which is $51.2 \%$. Hence, there is a gap of almost $6 \%$ in terms of influence coverage. It is also important to observe that in this dataset also exploitation of the community structure leads to more amount of influenced nodes for the proposed methodology. As an example, for $\mathcal{B}=26000$ under the uniform probability setting with $p_c=0.2$, the seed set selected by \textbf{BIMGT} and \textbf{BIMGTC} leads to $12964$ and $13212$ number of influenced nodes, which is almost $1.9 \%$ more.

\begin{figure*}[t]
\centering
\begin{tabular}{ccc}
\includegraphics[scale=0.2]{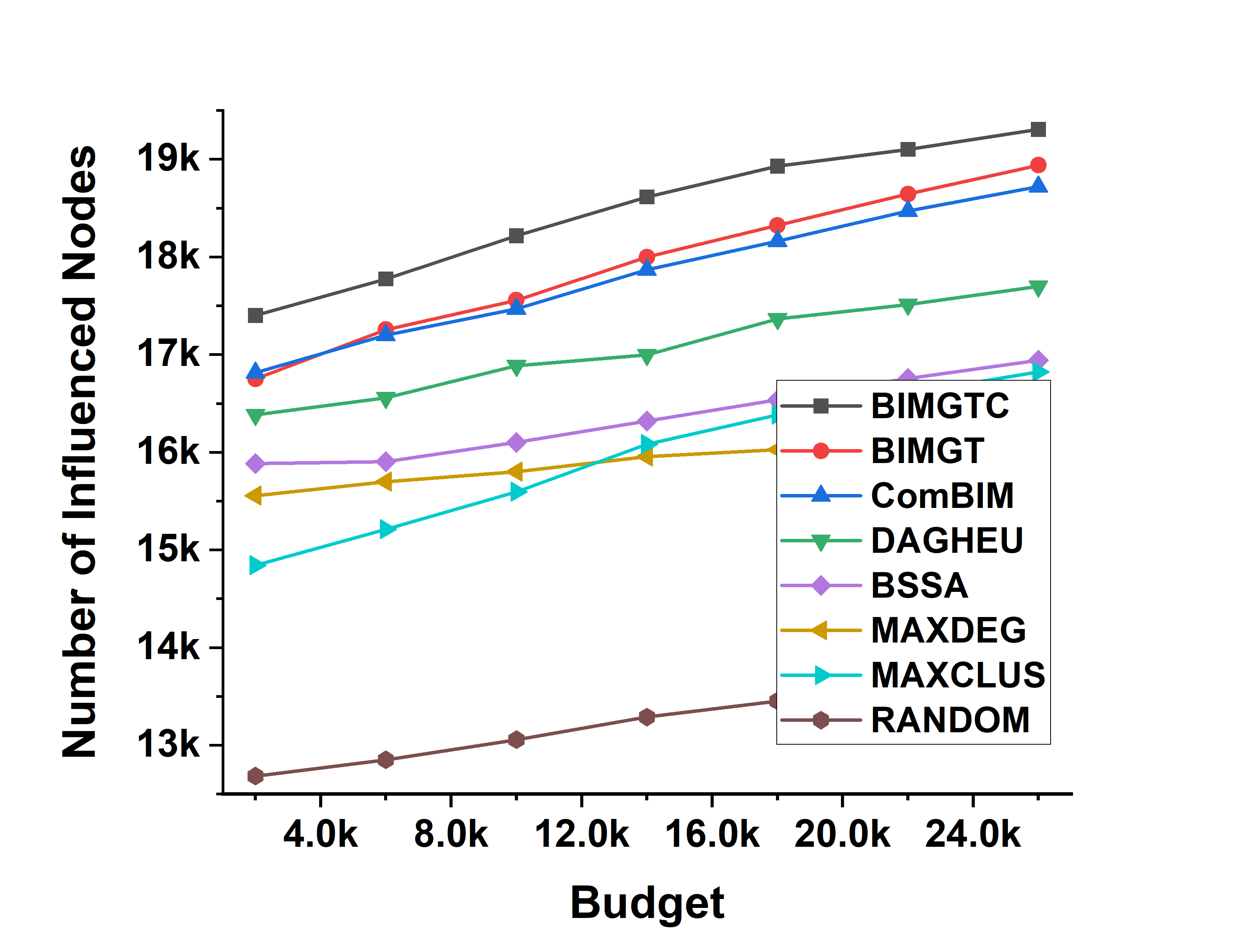} & \includegraphics[scale=0.2]{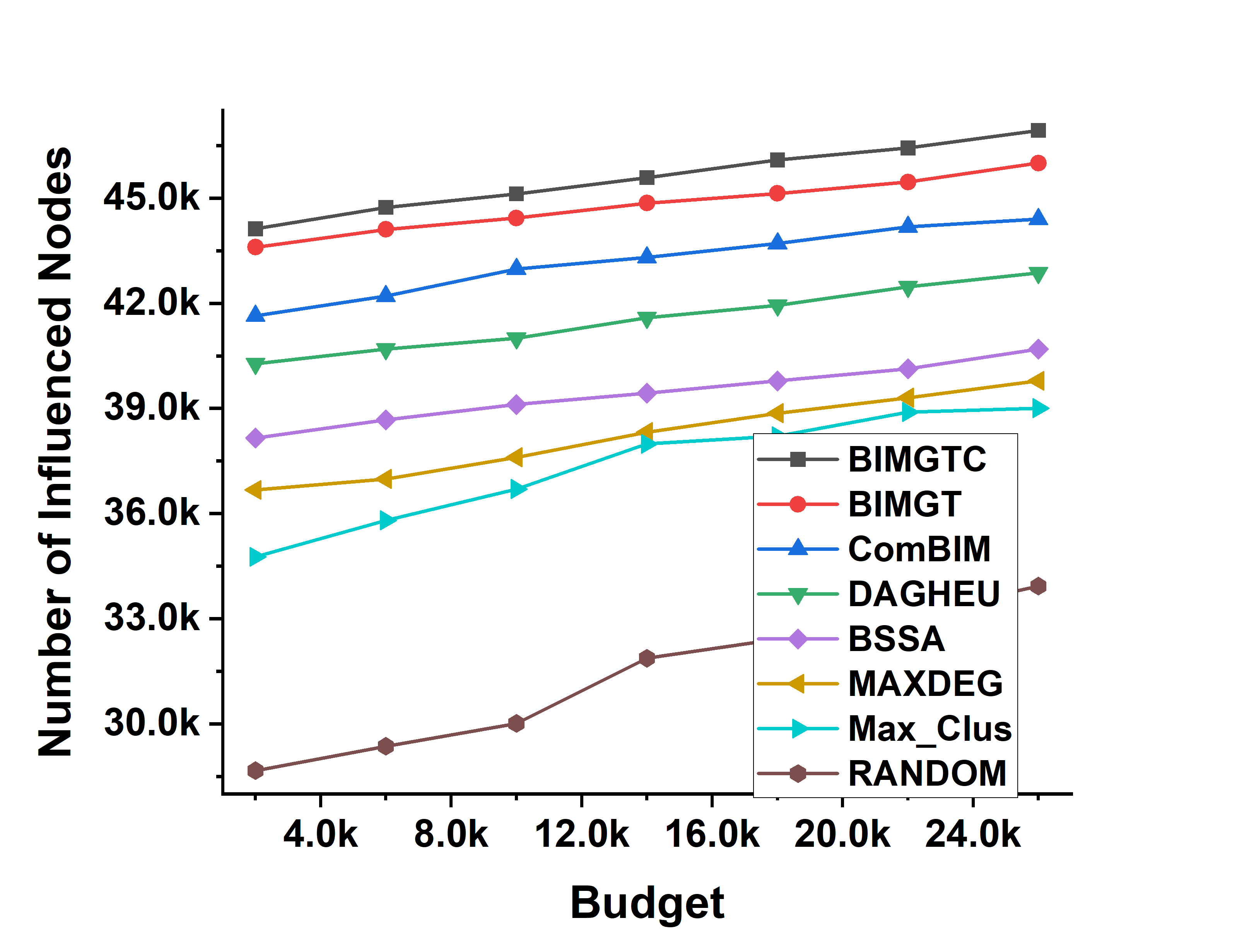} \\
(a) Epinions ($p_c=0.1$)  & (b) Epinions  ($p_c=0.2$) \\
\includegraphics[scale=0.2]{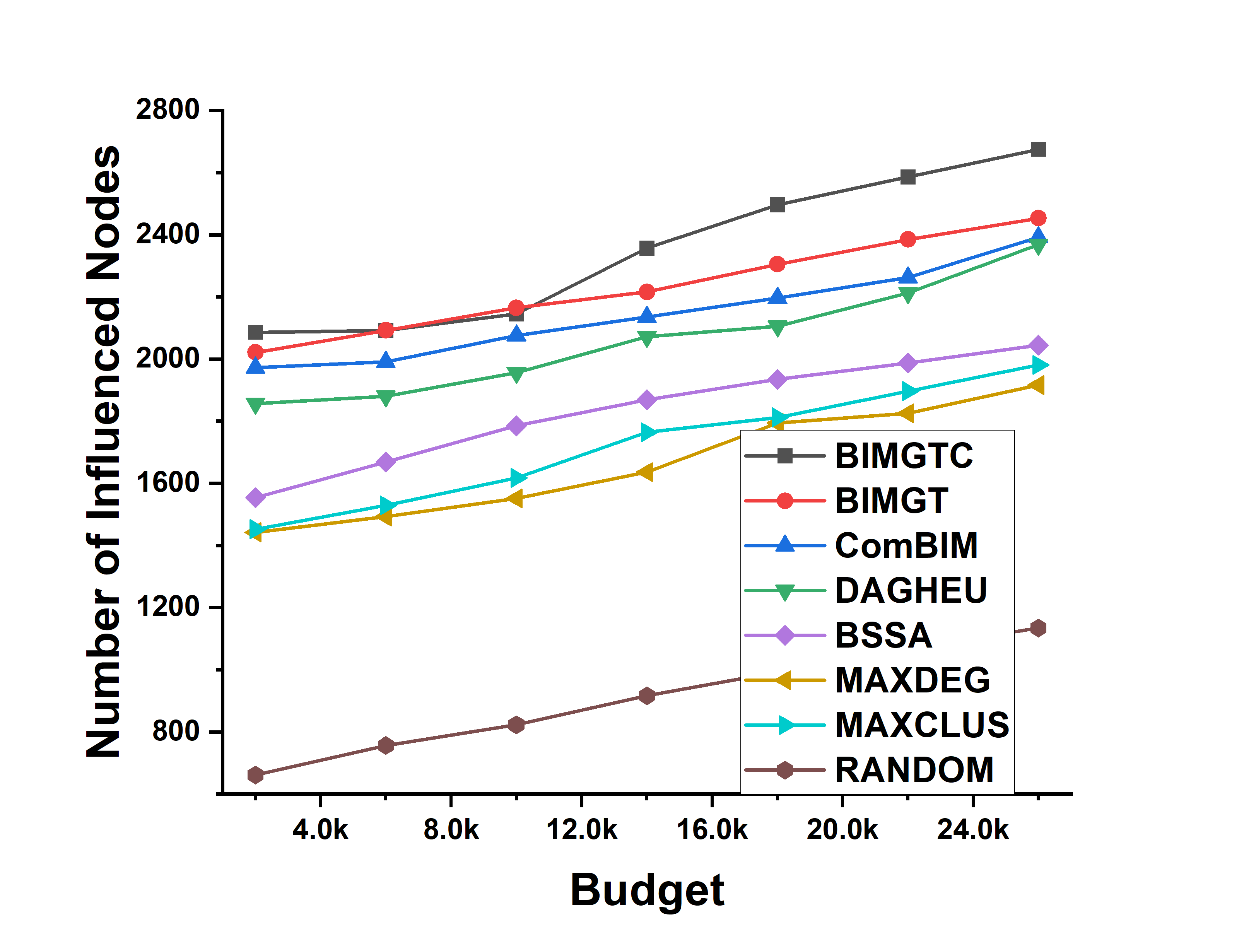} & \includegraphics[scale=0.2]{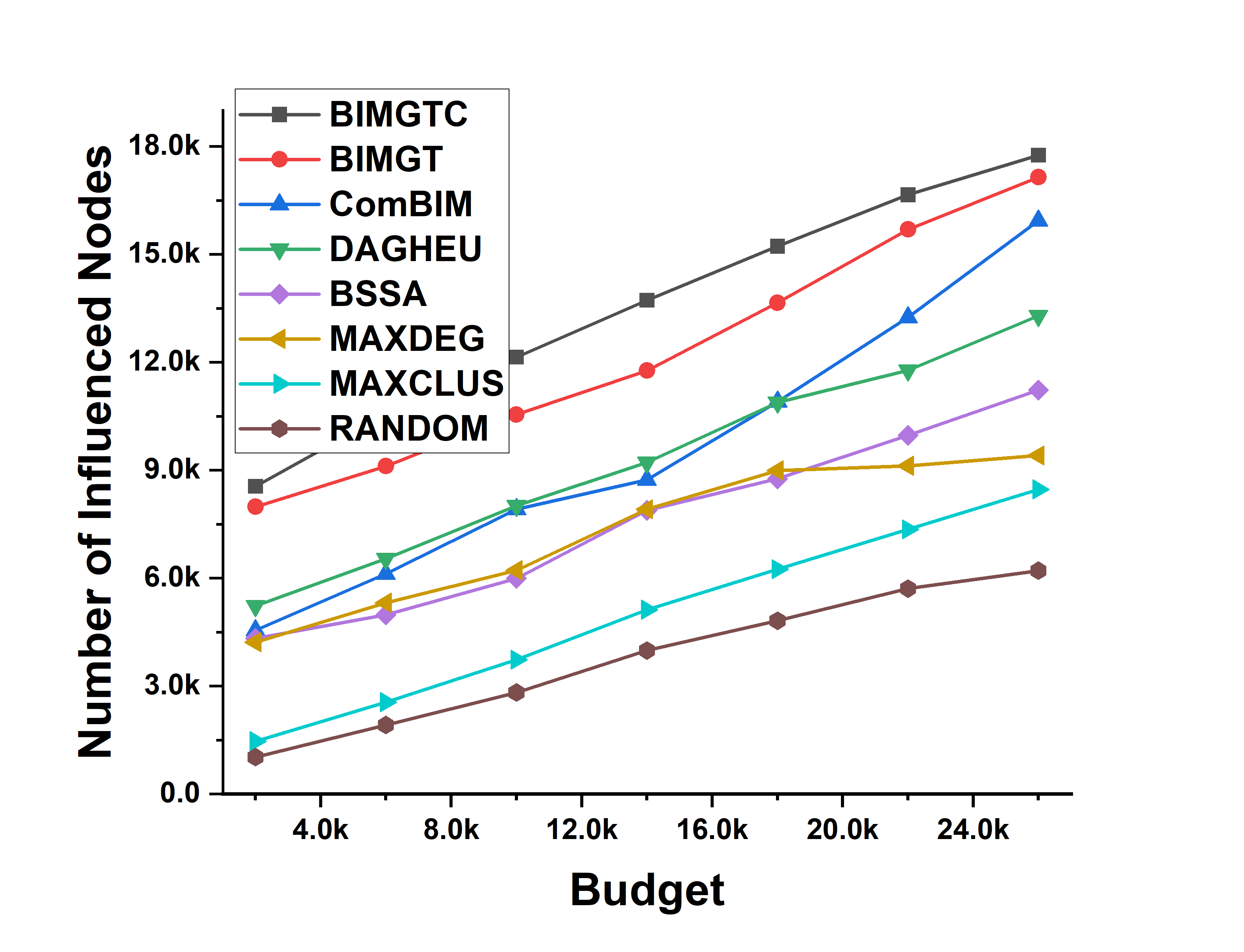} \\

(d) Epinions (Trivalancy) & (e) Epinions  (Weighted Cascade) \\
\end{tabular}
\caption{Budget Vs. Number of Influence Plots for the \textbf{Epinions}  dataset under the Uniform (with $p_c=0.1$ and $0.2$), Trivalency, Weighted Cascade and Count Probability settings.}
\label{Fig:Epinions}
\end{figure*}

Figure \ref{Fig:Epinions} shows the budget vs. number of influenced nodes plot for the Epinions dataset. Like previous two datasets, in this dataset also we observe that the seed set selected by the proposed methodologies leads to the more number of influenced nodes compared to the baseline methods. As an example in case of uniform probability setting when $p_c=0.2$ and $\mathcal{B}=26000$, among the proposed methodologies the seed set selected by the \textbf{BIMGTC} leads to the  maximum number of influenced nodes which is $46934$. On the other hand, among the existing  methods the seed set selected by \textbf{ComBIM} leads to the maximum number of influenced nodes which is $44402$. From Table \ref{Tab:Data_Stat}, it can be observed that the number of nodes of the Epinions dataset is $75879$. Hence, the percentage of nodes influenced by these seed nodes are $61.8 \%$ and $58.5 \%$, respectively. So, there is an approximate gap of $3.3 \%$ in terms of expected influence. It is also important to observe that due to the exploitation of the community structure of the input network, the number of influenced nodes increases. As an example when $p_c=0.2$ and $\mathcal{B}=26000$, the number of influenced nodes by the seed set selected by \textbf{BIMGTC} and \textbf{BIMGT} are $46934$ and $45997$, respectively. Next, we proceed to discuss computational time requirement.

\subsubsection{Computational Time} 
Table \ref{Tab:Time} shows the execution time of different algorithms for seed set selection. From this table it has been observed that the RANDOM takes the least amount of time across all the datasets. Next, MAXDEG takes more time than random as it needs to compute the degree of the nodes. MAXCLUS takes even more time than MAXDEG because computing clustering coefficient is much more computationally expensive operation than the degree. Remaining existing methods takes much more time than the baseline methods. It is important to observe that the computational time requirement of both the proposed methodologies are much less than that of the DAG-Based heuristic. It has been observed that the proposed methodologies takes more computational time compared to some of the baseline methods. However, it is important to realize that in many practical applications including viral marketing, computational advertisement etc., it is important to have an algorithm for seed set selection which has reasonable computational time with significant influence coverage. In this aspect, the proposed methodologies are far ahead compared to many existing methods.

\begin{center}
\begin{table}[h]
 \caption{Computational time requirement (in Secs.) for seed set selection for HEP\mbox{-}THOY, COND\mbox{-}MAT and Epinions Dataset}
    \label{Tab:Time}
\resizebox{0.9 \textwidth}{!}{ 
\begin{tabular}{ | c | c | c | c | c | c | c | c | c | c |}
    \hline
    \multirow{ 2}{*}{\textbf{Dataset}} & \multirow{ 2}{*}{\textbf{Budget}} &\multicolumn{8}{|c|}{\textbf{Algorithm}}\\
    \cline{3-10} 
     & & \textbf{BIMGTC} & \textbf{BIMGT} & \textbf{ComBIM} & \textbf{MAX\_DEG} & \textbf{MAX\_CLUS} & \textbf{RANDOM} & \textbf{BSSA} & \textbf{DAGHEU}  \\ \hline
  \multirow{ 7}{*}{HEP\mbox{-}THOY}  &2000 & 344 & 326 & 102 & 0.0464  & 0.3535  & 0.0032 & 22 & 1225 \\ \cline{2-10} 
    &6000 & 354 & 342 & 104  & 0.0488   & 0.3674  & 0.0085 & 26 & 1242  \\ \cline{2-10} 
    &10000 & 362 & 349& 104  & 0.0483   & 0.3666  & 0.0092 & 24 & 1251  \\ \cline{2-10} 
    &14000 & 364 & 348 & 104  & 0.0497  & 0.3556  & 0.0192  & 23 & 1258   \\ \cline{2-10} 
    &18000 & 372 & 352& 104  & 0.0492   & 0.3608  & 0.0269 & 25 & 1278  \\ \cline{2-10}
    &22000 & 389 &364 & 104  & 0.0503   & 0.3543  & 0.0333 & 28 & 1283  \\ \cline{2-10}
    &26000 & 395 & 368& 105  & 0.0496  & 0.3557  & 0.039  & 33 & 1299  \\ \cline{2-10}
    \hline
    
      \multirow{ 7}{*}{COND\mbox{-}MAT}   &   2000 & 810 & 785 & 345  & 0.0749 & 0.8999 & 0.0031 & 35 & 3640 \\ \cline{2-10} 
    &6000 & 812 & 796& 344  & 0.0758 & 0.8964 & 0.0119 & 36 & 3676  \\ \cline{2-10} 
    &10000 & 824 & 810& 344  & 0.0764 & 0.8968 & 0.0185 & 39 & 3697 \\ \cline{2-10} 
    &14000 & 835 & 826& 344  & 0.0772 & 0.8979 & 0.0311 & 38 & 3723  \\ \cline{2-10} 
    &18000 & 846 & 832& 349 & 0.0776  & 0.8981 & 0.0401 & 46 & 3761  \\ \cline{2-10} 
    &22000 & 858 & 844& 349 & 0.0784  & 0.9004 & 0.0450 & 42 & 3800  \\  \cline{2-10} 
    &26000 & 866 & 852 & 350 & 0.0778 & 0.9034  & 0.0541 & 44 & 3866  \\
    \hline
      \multirow{ 7}{*}{Epinions}  &    2000 & 1482 & 1456 & 1052 & 0.2649 &  199 & 0.0081 & 66 & 14620 \\ \cline{2-10} 
   & 6000 & 1521 & 1490& 1051 & 0.2758 & 201 & 0.0219 & 72 & 14751 \\ \cline{2-10} 
   &10000 & 1532 & 1506 & 1082 & 0.3821 & 226 & 0.0234 & 75 & 14762 \\ \cline{2-10} 
   & 14000 & 1546 & 1523& 1079 & 0.3422 & 225 & 0.0312 & 78 & 14463  \\ \cline{2-10} 
   & 18000 & 1562 & 1540& 1116 & 0.3776 & 238 & 0.0427 & 81 & 14561 \\ \cline{2-10} 
   & 22000 & 1575 & 1556& 1131 & 0.5261 & 231 & 0.0392 & 79 & 14600 \\ \cline{2-10} 
   & 26000 & 1596 & 1574& 1139 & 0.5116 & 229 & 0.0798 & 92 & 14786 \\
    \hline
    \end{tabular}
    }
\end{table}
\end{center}

\section{Conclusion and Future Direction} \label{Sec:CFD}
In this paper, we have proposed a co\mbox{-}operative game theoretic framework for the budgeted influence maximization problem. Particularly, we formulate a co\mbox{-}operative game, where the users of the network are the players, and for any subset of the players, their utility is defined as the expected influence by the users of the subset under the MIA Model of diffusion. We have used the solution concept called `shapley value' and proposed an iterative algorithm for selecting influential users in the network. We have also shown that the proposed methodology can select better quality seed set if the community structure of the network is exploited. Experiments with real\mbox{-}world social network datasets demonstrate the superiority of the proposed methodologies. Now, this study can be extended in different directions. First of all, we have not given any approximation guarantee of our proposed methodologies with respect to an optimal seed set. It will be interesting to come up with some worst case performance guarantee for our proposed methodologies. We have not considered the time\mbox{-}varying nature of the social network. There are many other solution concepts of co\mbox{-}operative game, such as Banzhaf index \cite{dubey1979mathematical} etc. These solution concepts can be used instead of `shapley value' and compare the performance with the proposed methodologies.
\bibliography{paper}

\end{document}